\documentclass[12pt]{amsart}
\usepackage{amssymb}
\usepackage{amsthm}
\usepackage{subfigure}
\usepackage{graphicx}
\usepackage{latexsym}
\usepackage{color}
\usepackage{amstext,amssymb, amsfonts}

\newcommand{\Rst}{\mathbb{R}}
\newcommand{\Zst}{\mathbb{Z}}
\newcommand{\Nst}{\mathbb{N}}
\newcommand{\Rdst}{\mathbb{R}^d}

\newcommand{\set}[2]{\big\{ \, #1 \, \big| \, #2 \, \big\}}
\newcommand{\norm}[1]{\lVert#1\rVert}

\def\supp{\operatorname{supp}}
\def\sinc{\operatorname{sinc}}

\def\clD{{\mathcal{M}}}
\def\clP{{\mathcal{P}}}

\def\bI{{\mathbb I}}

\def\bom{\mathcal{B}_\Omega}

\def\pd{path density }

\def\nom{\textit{Nyq}_\Omega}
\def\nomab{\textit{Nyq}_\Omega^{A,B}}
\def\Hom{\textit{Hom}_\Omega}
\def\Par{\textit{Par}_\Omega}

\DeclareMathOperator{\littleo}{o}

\newcommand{\length}{\ell}
\newcommand{\SP}{Q}

\newcommand{\sett}[1]{\ensuremath{\left \{ #1 \right \}}}
\newcommand{\abs}[1]{\ensuremath{\left| #1 \right| }}
\newcommand{\ip}[2]{\ensuremath{\left<#1,#2\right>}}
\newtheorem{definition}{Definition}[section]
\newtheorem{thm}{Theorem}[section]
\newtheorem{lemma}[thm]{Lemma}
\newtheorem{proposition}[thm]{Proposition}
\newtheorem*{proposition*}{Proposition}

\newtheorem{remark}[thm]{Remark}
\def\lrd{L^2(\mathbb{R}^d)}

\newcommand{\rd}{\mathbb{R}^d}

\newcommand{\bound}{{\frac{1}{2}+ \frac{d \pi^{2d-2}}{2^{2d-1}}\frac{B}{A}}}

\newcounter{con}
\newenvironment{condition}{\begin{list}{{ \bf (C\arabic{con}) \ }}{\usecounter{con}
\setlength{\leftmargin}{14pt}
\setlength{\rightmargin}{12pt}
\setlength{\itemindent}{-25 pt}
\setlength{\labelsep}{-3 pt}
\setlength{\itemsep}{5 pt}
}}{\end{list}}

\def\Con#1{\textbf{(C\ref{con:#1})}}

\newcounter{conAB}
\newenvironment{conditionAB}{\begin{list}{{ \bf (C\arabic{conAB}$^{A,B}$) \ }}{\usecounter{conAB}
\setlength{\leftmargin}{14pt}
\setlength{\rightmargin}{12pt}
\setlength{\itemindent}{-25 pt}
\setlength{\labelsep}{-3 pt}
\setlength{\itemsep}{5 pt}
}}{\end{list}}

\def\Re{\Rst}

\newlength{\noteWidth}
\title{On Minimal Trajectories for Mobile Sampling of Bandlimited Fields}
\author[K. Gr\"{o}chenig]{Karlheinz Gr\"{o}chenig}
\address{Faculty of Mathematics, University of Vienna, Austria}
\email{karlheinz.groechenig@univie.ac.at }
\author[J. L. Romero]{Jos\'{e} Luis Romero}
\email{jose.luis.romero@univie.ac.at }
\author[J. Unnikrishnan]{Jayakrishnan Unnikrishnan}
\address{Audiovisual Communications Laboratory,
School of Computer and Communication Sciences,
Ecole Polytechnique F\'{e}d\'{e}rale de Lausanne (EPFL),
Switzerland}
\email{jay.unnikrishnan@epfl.ch}
\author[M. Vetterli]{Martin Vetterli}
\email{martin.vetterli@epfl.ch}



\begin{document}
\begin{abstract}
We study the design of sampling trajectories for  stable sampling and
the
reconstruction of bandlimited spatial fields using
mobile sensors. The spectrum  is assumed to be a symmetric convex set.
As a performance metric we use the path density of the
set of sampling trajectories that is  defined as the
total distance traveled by the moving sensors per unit spatial volume of the spatial region being monitored. Focussing
first on parallel lines, we identify the set of parallel lines with minimal path density that contains a set of stable
sampling for fields bandlimited to a known set. We then show that the problem becomes ill-posed when the optimization is
performed over all trajectories by demonstrating a feasible trajectory set with arbitrarily low path density. However,
the problem becomes well-posed if we explicitly specify the stability margins. We demonstrate this by obtaining a
non-trivial lower bound on the path density of an arbitrary  set of
trajectories that contain a sampling set with explicitly
specified stability bounds.
\end{abstract}
\maketitle

\section{Introduction}
The reconstruction of a function from given measurements is a
fundamental task in data processing and occupies numerous directions
of research in mathematics and engineering. A typical problem
requires the reconstruction or approximation of a physical field from
pointwise measurements. A field may be a distribution of temperatures
or water pollution or a solution to a diffusion equation, in
mathematical terminology a field is simply a smooth function of
several variables. The standard assumption on the smoothness is that
the field is bandlimited to a compact spectrum. If  the spectrum is
a  fundamental domain of a lattice in $\rd$ or a symmetric convex
polygon in $\mathbb{R}^2$, then there exist precise reconstruction formulas from
sufficiently many samples in analogy to the Shannon-Whittaker-Kotelnikov sampling theorem~\cite{LR00,petmid62}.

Let
\begin{equation}
\widehat{f}(\omega) = \int_{\Re^d} f(r) e^{-{2\pi \sf i}\langle
  \omega, r\rangle } dr, \qquad \omega \in \Re^d \, ,
\label{eqn:fourtran}
\end{equation}
be the Fourier transform~\footnote{Note that in  \cite{unnvet11Allerton} and
\cite{unnvet13TIT} the Fourier transform was defined without the $2 \pi$ in the exponent.}
 of $f\in L^1(\rd ) $ or $f\in \lrd $, where
${\sf i}$ denotes the imaginary unit  and $\langle \omega ,r
\rangle$ denotes the scalar product between vectors $\omega $
and $r$ in $\Re^d$. We say that $f$ is bandlimited to the  closed set
$\Omega \subset \Re^d$, if its  Fourier transform $\hat{f}$ is supported on $\Omega$. In this case we write
\begin{equation}
\bom: = \{f \in L^2(\Re^d): \widehat{f}(\omega) = 0 \mbox{ for almost every } \omega \notin \Omega\}\label{eqn:bmodefn}
\end{equation}
for  the space  of fields
with finite energy bandlimited to the spectrum  $\Omega$. In the
context of field estimation we always assume that the spectrum is a
compact, symmetric, convex set.

The classical theory of sampling and reconstructing of such high-dimensional bandlimited fields dates back to Petersen
and Middleton \cite{petmid62} in signal analysis and to
Beurling~\cite{beurling66} in harmonic analysis. Both  identified conditions for reconstructing such
fields from their point
measurements in $\Re ^d$.
Further research on non-uniform  sampling generated more results on conditions for perfect
reconstruction from samples taken at non-uniformly distributed spatial
locations. See ~\cite{Gro92b,GS01}  and the  survey \cite{aldgro01}.
Previous  work deals primarily   with the problem of reconstructing
the field from measurements
taken by a collection of static sensors distributed in space, like
that shown in Figure \ref{fig:sampR2a}. In this case the
performance metric for  quantifying  the efficiency of a sampling
scheme is the spatial density of samples. This
is  the average  number of sensors  per unit volume  required for
the stable sampling of the monitored
region.

\begin{figure}
\centering
\hspace*{\fill}
\subfigure[Static sampling on points]{
\includegraphics[width=2.1in]
{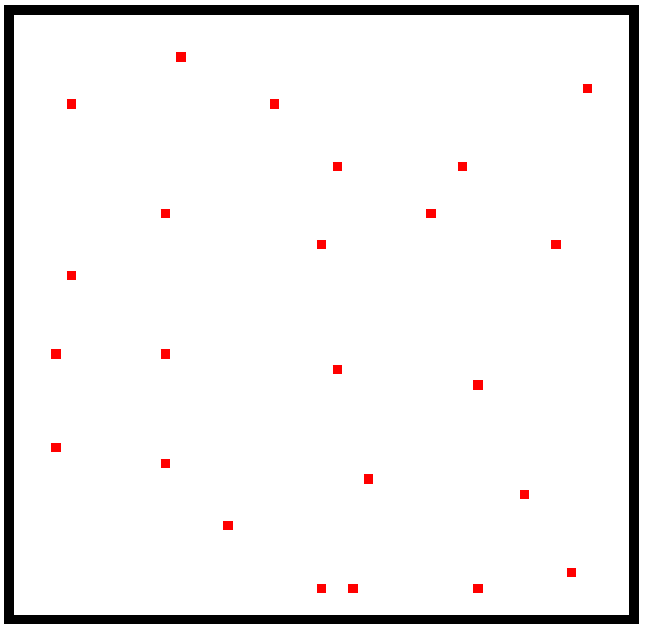}
\label{fig:sampR2a}
}
\hfill
\subfigure[Mobile sampling on a curve]{
\includegraphics[width=2.1in]
{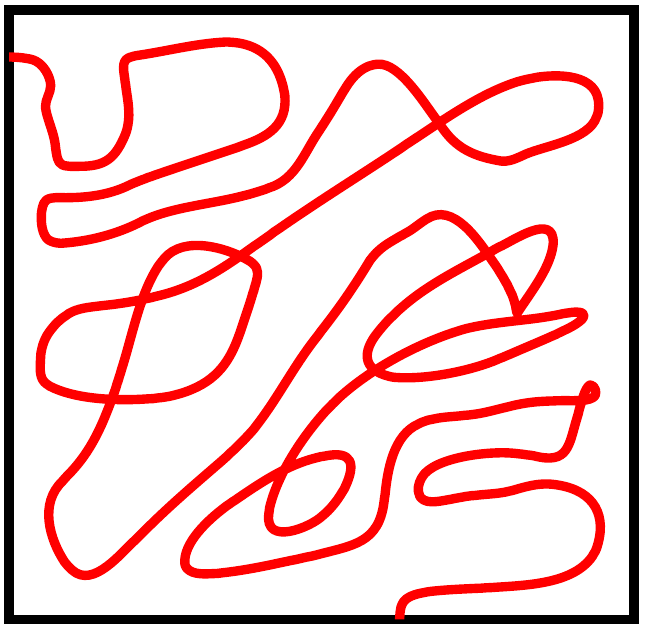}
\label{fig:sampR2b}
}
\hspace*{\fill}
\caption[Two approaches for sampling a field in $\Re^2$]{Two approaches for sampling a field in $\Re^2$}
\end{figure}

In this paper we investigate a different method for the acquisition
of the samples, which we call \emph{mobile sampling}. The samples are
taken by a mobile sensor that moves  along a continuous path,  as is
shown in Figure \ref{fig:sampR2b}. In such a case  it is often
relatively inexpensive to increase the spatial
sampling rate along the sensor's path while the main cost of the sampling scheme comes from the total distance that
needs to be traveled by the moving sensor. Hence it is reasonable to assume that the sensor can record the field values
at an arbitrarily high but finite resolution on its path.

The new method for the acquisition of samples changes the mathematical
nature of the problem completely. When using samples from static
sensors, we need to establish a sampling inequality with
evaluations of the form
\begin{equation}
A \|f\|^2 \leq \sum_{\lambda \in \Lambda} |f(\lambda)|^2 \leq B
\|f\|^2, \quad \mbox{ for all }f \in \bom \, ,
\label{eqn:frameboundsa}
\end{equation}
for constants $A,B>0$ independent of $f$.

For mobile sampling, we need to establish a ``continuous'' sampling  inequality of the form
\begin{equation}
  \label{eq:c2}
A \|f\|^2 \leq \int_{\rd }  |f(r)|^2 \, d\mu (r) \leq B \|f\|^2, \mbox{ for
  all }f \in \bom \, ,
  \end{equation}
where $\mu $ is the sum of line integrals along the paths.

Again, the performance metric should reflect the cost required
for the data acquisition. For \eqref{eqn:frameboundsa} the appropriate
metric is the average number of sensors, i.e., samples, per unit
volume. For \eqref{eq:c2} some of us have argued
in~\cite{unnvet11Allerton} and in \cite{unnvet13TIT} that  the
relevant metric  should be the average path length  traveled by the
sensors  per unit  volume (or area, if  $d =2$). We call this metric
the \textit{path density}. Such a metric is directly  relevant in applications
like environmental monitoring using moving sensors
\cite{unnvet13TSP}, \cite{sinnowram06}.
In retrospect this metric is also useful in designing $k$-space trajectories for Magnetic
Resonance Imaging (MRI) \cite{benwu00}, where the path density
can be used as a proxy for the total scanning time per unit area in $k$-space.

The continuous sampling inequality~\eqref{eq:c2} raises many
interesting questions both for engineers and for mathematicians. On
the   mathematical side are  the abstract construction of continuous
frames in the sense of ~\cite[Chaps.~3 and 5]{AliAnt2000} or~\cite{FR05} or the analysis of
sampling measures and their properties, see~\cite{Lue00,OC98} for
a theory of sampling measures for Fock spaces and for Bergmann spaces.
On the engineering  side, we need to  design  concrete, realizable  trajectories with
a small path density for bandlimited fields with convex spectrum. This
problem was introduced
by some of us  in \cite{unnvet11Allerton} and \cite{unnvet13TIT} and
answered for the special case of trajectory sets that consist   of a
union of uniformly spaced  lines.

The contribution in this article is twofold.
 First, we study arbitrary trajectory sets of parallel lines and
 derive a necessary condition for the minimal path density in the
 style of Landau's  famous result  in ~\cite{lan67}. Extending  the results in
\cite{unnvet11Allerton,unnvet13TIT, unnvet13SAMPTA} we  show, in Theorem \ref{th:parallel}, that the minimal path
density achievable by sampling along  trajectories
 of arbitrary parallel lines
 is  exactly   the area of the maximal hyperplane section of the spectrum. We work under the standard assumption
that the spectrum of the signals is convex and symmetric (although some results
holds for more general spectra, see Section \ref{sec:conc}).

At first glance, the sampling along parallel lines seems to be an easy
generalization of point sampling, because it can be reduced to the  sampling problem in
smaller dimensions.   However, even this case offers some interesting and challenging
problems that we did not envision before. For instance, in Section~3
we will use the existence of universal samplings sets as established
by Olevsky and Ulanovsky~\cite{olul08} and by Matei and Meyer~\cite{mame10} in
order  to prove that the frame bounds  are uniform for sections of
convex sets. In addition,  this result enriches our knowledge about the
properties of universal sampling sets. For another  crucial argument  we need
the Brunn-Minkowski inequality~\cite{ga02}. 

Of course, the mathematician's immediate instinct is to study more
general sets of trajectories and try to prove a result analogous to
Landau's necessary condition for the path density. We show in
Proposition~\ref{prop:infzero} that such a
result cannot hold by constructing stable trajectory sets  with arbitrarily
small path density. Thus in a sense there is no optimal configuration
of paths and the problem of optimizing the path density is
ill-posed.  This answers a question raised in
\cite{unnvet13TIT}. However, as soon as we minimize over trajectory
sets with given stability parameters $A,B$ (uniform frame bounds) the
optimization problem becomes well-posed. Our main density  result
(Theorem~\ref{th:infab})  shows  that  the path density for a stable set of 
trajectories  is bounded below by an expression involving the
stability parameters and the geometry of the spectrum.

This is a report on a successful and fertile collaboration between
engineers and mathematicians. We, the mathematicians, are intrigued by
the questions that motivate mobile sensing. Although the mathematical
literature has investigated generalizations of sampling (the theorems
of Sereda-Logvinenko and the theory of sampling measures) for the sake
of generalization, we would never have dreamt of the particular
conditions on the paths that are imposed by practical considerations
(see condition ~\Con{path3}). We, the engineers, are intrigued by the
mathematical subtleties that popped up at every corner and
subsequently led to an extended theory of path sampling.

The paper is organized as follows. In Section \ref{sec:probstat} we describe the formal problem statement. Then, in
Section \ref{sec:optresults}, we characterize the minimal density of sampling trajectories consisting of
parallel lines. Section \ref{sec:disc} treats the problem of optimizing over arbitrary trajectories, and
Section \ref{sec:conc} presents some conclusions. The proofs of some technical lemmas needed throughout the article are
postponed to Section \ref{sec_tec}, so as not to obstruct the flow of the article.

\bigskip

\noindent {\bf Notation}. We use $\langle \cdot  , \cdot  \rangle$ to
denote the canonical inner product on $\rd $ and $\lrd $, and $e_k$ to denote the
unit vector along the $k$-th coordinate axis. For $u \in \Re^d$ we
denote the hyperplane orthogonal to $u$ through the
origin by $u^\perp = \{x\in\Re^d: \langle x, u \rangle = 0\}$,  and
$\clP_{u^\perp}S$ denotes the orthogonal projection of a set
$S\subset \Re ^d$ onto the
hyperplane $u^\perp$. For a
set $S \subset \Re^d$ we use $|S|$ to denote the
volume of $S$ with respect to Lebesgue measure.
 By  $B_a^d(x)$ we denote the  closed  Euclidean ball of radius $a$ centered at
 $x \in \Re^d$  and $B_a^d = B_a^d(0)$, and by $Q_a(x) =
 x+[-a,a]^d$ we denote the cube of width $2a$ centered at $x$.
The cardinality of a finite set $\Lambda$ is denoted by  $\# \Lambda $.

Let  $I= [a,b] $ be a bounded   interval and  $\gamma : I \to \rd $ is a
curve in $\rd $. We say that $\gamma $ is  rectifiable, if  $\ell
(\gamma ) = \sup \sum _{k=0} ^{n-1} 
\|\gamma (t_{k+1}) - \gamma (t_k)\| $ is finite, where the supremum is
taken over all finite partitions $a = t_0 < t_1 < \dots < t_{n-1}<t_n
=b $. In this case  $\ell (\gamma )$ is called  the arc length of
$\gamma$. Every piecewise differentiable curve is rectifiable. 

If
quantities $X$ and $Y$ satisfy the condition that there exist $A,B > 0$ with
\[
A Y \leq X \leq B Y,
\]
we write  $X \asymp Y$. We also use the notation $X \lesssim Y$ to indicate that there exists
$B>0$ such that $X \leq B Y$.
To compare the size of functions $g,h :\Re \to \Re ^+$, we use the
Landau notation $O$ and $o$. The symbol
 $h = O(g(x))$ means that  there exist $k>0$ and $y$ such that for all
 $x > y$, we have $|h(x)| \leq k g(x)$, and  $h = o(g(x))$ if $\lim
 \frac{h(x)}{g(x)} = 0$.

We say that a set of points $\Lambda \subset \Re^d$ is \emph{uniformly
  discrete} or \emph{separated} if  $\inf\{\|x-y\|: x,y \in 
\Lambda, x \neq y\} > 0$, i.e., there exists $r > 0$ such that for any two distinct points $x,y \in \Lambda$ we have
$\|x-y\| > r$. For example, a  lattice  in $\Re^d$ is  uniformly discrete, but a sequence in $\Re^d$ converging to
a point in $\Re^d$ is not. The lower and upper Beurling densities of $\Lambda \subseteq \Rdst$ are
\begin{align*}
D^{-}(\Lambda) := \liminf_{a \to \infty} \inf_{x \in \Re^d}\frac{\#(\Lambda \cap B_a^d(x))}{\abs{B_a^d}},
\\
D^{+}(\Lambda) := \limsup_{a \to \infty} \sup_{x \in \Re^d}\frac{\#(\Lambda \cap B_a^d(x))}{\abs{B_a^d}}.
\end{align*}
 For every  compact set $K \subseteq \Rdst$ with non-empty interior and whose boundary has
 measure zero,
the lower density can be also calculated as:
\[
D^{-}(\Lambda) = \liminf_{a \to \infty} \inf_{x \in \Re^d}\frac{\#(\Lambda \cap (aK+x))}{a^n\abs{K}},
\]
and a similar statement holds for the upper density~\cite[Lemma 4]{lan67}.

The \emph{covering constant} of a set $\Lambda \subseteq \Rdst$ is
\begin{align*}
\nu = \nu (\Lambda )  = \sup_{x\in \rd } \# (\Lambda  \cap  Q_{1/2}(x) ).
\end{align*}
A set is called \emph{relatively separated} if it has a finite
covering constant, which  holds if and only if it  has finite upper Beurling density.

A set $E \subseteq \Rdst$ is called a \emph{convex body} if it is convex, compact and has non-empty interior. A convex
body is called \emph{centered} if $0 \in E^\circ$ and symmetric if $E=-E$. The following fact will be frequently used
in approximation arguments.

\begin{lemma}[Dilation of centered convex bodies]
\label{lemma_conv_1}
Let $E \subseteq \Rdst$ be a centered convex body. Let $\delta\in(0,1)$, then $E \subseteq
(1+\delta)E^\circ$ and $(1-\delta)E \subseteq E^\circ$.
\end{lemma}

\section{Trajectory sets and sampling} \label{sec:probstat}
A \textit{trajectory} $p$ in $\Re^d$ is  the image of
curve $\gamma : \Re  \to  \Re^d$, i.e., $p = \gamma (\Re )$ such that
the restriction of $\gamma $ to any finite interval is rectifiable.
A \textit{trajectory set} $P$ is defined as a countable collection of trajectories:
\begin{equation}
P = \{p_i: i\in \bI\} \label{eqn:trajset} \, ,
\end{equation}
where $\bI$ is a countable set of indices and every  $p_i$ is a
trajectory.
In analogy to the Beurling density we define the lower and upper
  \textit{\pd} of a trajectory set $P$  as follows:
  \begin{definition}
    Let  $\clD^P(a,x)$ be  the total arc-length of the trajectories in
   $P \cap  B_a^d(x)$. Then the \emph{lower path density} and the
   upper path density are 
\begin{align}
\label{eqn:pddefnlow}
\ell^{-}(P) := \liminf_{a \to \infty} \inf_{x\in\Re^d} \frac{\clD^P(a,x)}{|B_a^d|},
\\
\label{eqn:pddefnup}
\ell^{+}(P) := \limsup_{a \to \infty} \sup_{x\in\Re^d} \frac{\clD^P(a,x)}{|B_a^d|}.
\end{align}
If $\ell^{-}(P)=\ell^{+}(P)$, then $P$  is said to possess the
homogeneous path density  $\ell (P) = \ell ^{\pm} (P) $.
\end{definition}
\begin{figure}
\centering
\subfigure[The Beurling density of a lattice used in classical
sampling theory quantifies the number of samples per unit area (or per
volume  in higher dimensions).]{ 
\includegraphics[width=2.6in]
{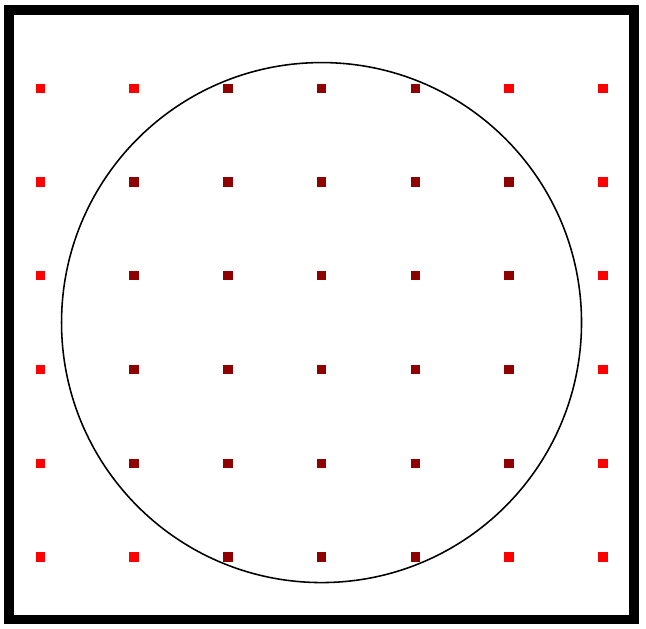}
\label{fig:beur}
}\hfill
\subfigure[The path density of a trajectory set used in mobile
sampling quantifies the total length of the paths per unit area (or per
volume  in higher dimensions).]{ 
\includegraphics[width=2.6in]
{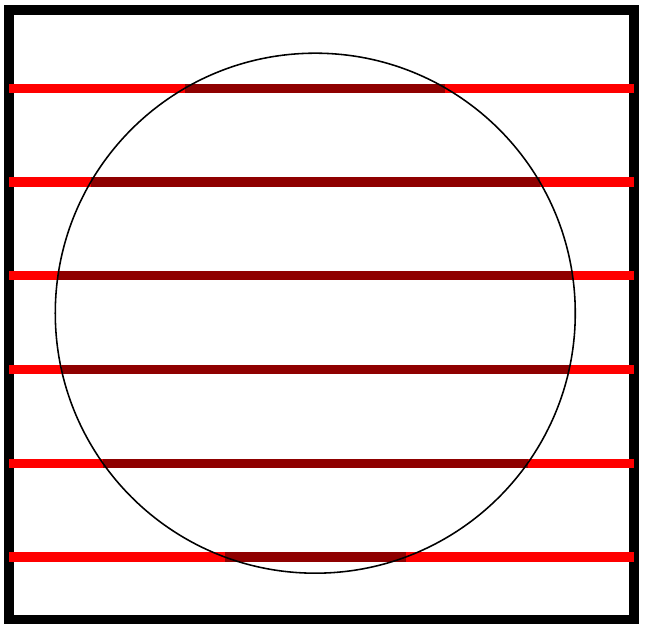}
\label{fig:pd}
}
\caption{Illustration of Beurling and path densities for classical sampling on a lattice and mobile sampling on a set of
equispaced parallel lines (uniform set) in $\Re^2$.}
\label{fig:beurlingandpd}
\end{figure}  
An illustration comparing Beurling and path densities is provided in Figure \ref{fig:beurlingandpd}. 
As with Beurling's density, the path density does not depend on the particular choice of the
Euclidean ball. More precisely, we have the following result.
\begin{lemma}
\label{lemma_dens_tile}
Let $K\subset \rd$ be a compact set with non-empty interior and with a boundary of
measure zero and let $\clD^P(a,x,K)$ be the total arc-length of trajectories from $P$ located in
$x+aK$. Then $\ell^{-}(P) := \liminf_{a \to \infty} \inf_{x\in\Re^d}
\frac{\clD^P(a,x,K)}{a^n\abs{K}}$ and $\ell^{+}(P) := \limsup_{a \to \infty} \sup_{x\in\Re^d}
\frac{\clD^P(a,x,K)}{a^n\abs{K}}$.
\end{lemma}
Lemma \ref{lemma_dens_tile} can be proved by following Landau's proof of the analogous result for Beurling's
density \cite[Lemma 4]{lan67}. We refer the reader to that article.

The simplest example of a trajectory set in $\Re^2$ is
a sequence of equispaced parallel lines in  $\Re^2$ (e.g., see Figure \ref{fig:unifset2d}). We call
such a trajectory set a \textit{uniform set in $\Re^2$}. Such a uniform set has a path density equal to
$\frac{1}{\Delta}$, where $\Delta$ is the spacing between the lines
(see \cite[Lemma 2.2]{unnvet13TIT}). Similarly a
\textit{uniform set in $\Re^d$} is defined as a collection of parallel lines in $\Re^d$ such that the cross-section
forms a $(d-1)$-dimensional lattice, see Figure \ref{fig:unifset3d}.
\begin{figure}[ht]
\centering
 \hspace*{\fill}
\subfigure[Uniform set in $\Re^2$]{
\includegraphics[width=2.1in]
{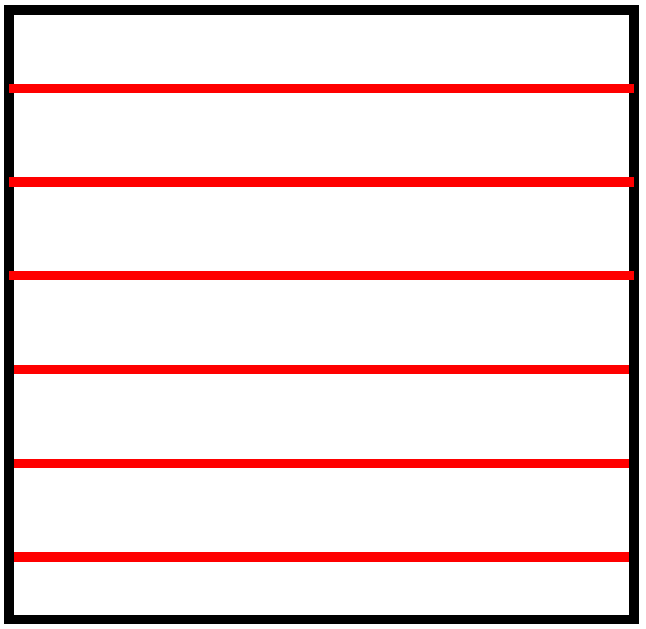}
\label{fig:unifset2d}
}
\hfill
\subfigure[Uniform set in $\Re^3$]{
\includegraphics[width=2.1in]
{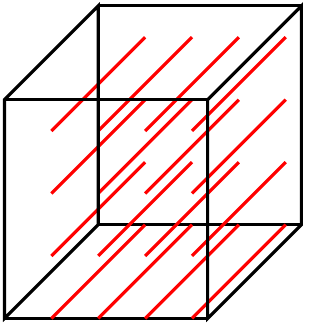}
\label{fig:unifset3d}
}
\hspace*{\fill}
\caption{Examples of uniform sets in $\Re^2$ and $\Re^3$.}
\end{figure}
Recall that for static sampling with fixed sensors the appropriate notion of
stability was the sampling inequality~\eqref{eqn:frameboundsa}. 
For mobile sampling along trajectory sets we require similar
conditions for the stability and are led to the following
definition.

\begin{definition} \label{trajset}
A trajectory set $P$ of the form (\ref{eqn:trajset}) is called a \emph{stable Nyquist trajectory set} for $\bom$ if $P$ 
satisfies the following conditions:
\begin{quote}
\nobreak
\begin{condition}

\item \label{con:recon2} \emph{[Nyquist]} There exists a uniformly discrete set $\Lambda$ of points on the trajectories
in $P$, $\Lambda \subset P$,   such that $\Lambda$ forms a set of
stable sampling for $\bom$.

\item \label{con:path3} \emph{[Non-degeneracy]} There exists a function $\delta : \Re ^+ \to \Re ^+$ such that
  $\delta(a) = \littleo(a^d)$ with the following property:  For every  $x \in
\Re^d$ and every $a \gg 1$, there is a rectifiable curve
$\alpha:[0,1] \to \Rdst, $ (depending on $x$ and $a$),  such that (i)
$\length(\alpha)=\clD^P(a,x) +\delta(a)$, and
\mbox{(ii) $P \cap B_a^d(x) \subset \alpha ([0,1])$}, i.e., the curve $\alpha$ contains the
portion of the trajectory set $P$ that is located within $B_a^d(x)$.

\end{condition}
\end{quote}
\end{definition}
\noindent For brevity, we will denote  the collection of all stable Nyquist
trajectory sets for $\Omega$ by  $\nom$.

Condition~\Con{path3} is a regularity condition motivated by the model of mobile sensors.
It ensures that for all $x \in \Re^d$ a single sensor moving along a rectifiable curve with total length $\clD^P(a,x)
+\delta(a)$ can cover the portions of the trajectories in the ball $B_a^d(x)$.
An illustration of such a curve for the trajectory set in Figure
\ref{fig:pd} is shown in Figure \ref{fig:c2condn}. 
Thus although there may be a countable collection of paths in $P$, a single sensor can be used to cover the portions of
$P$ inside $B_a^d(x)$, without affecting the total distance traveled per unit area.
This means, in particular, that the path density does indeed capture the total distance per unit area covered by a
single moving sensor using the trajectories in $P$.
\begin{figure}
\centering
\includegraphics[width=2.5in]{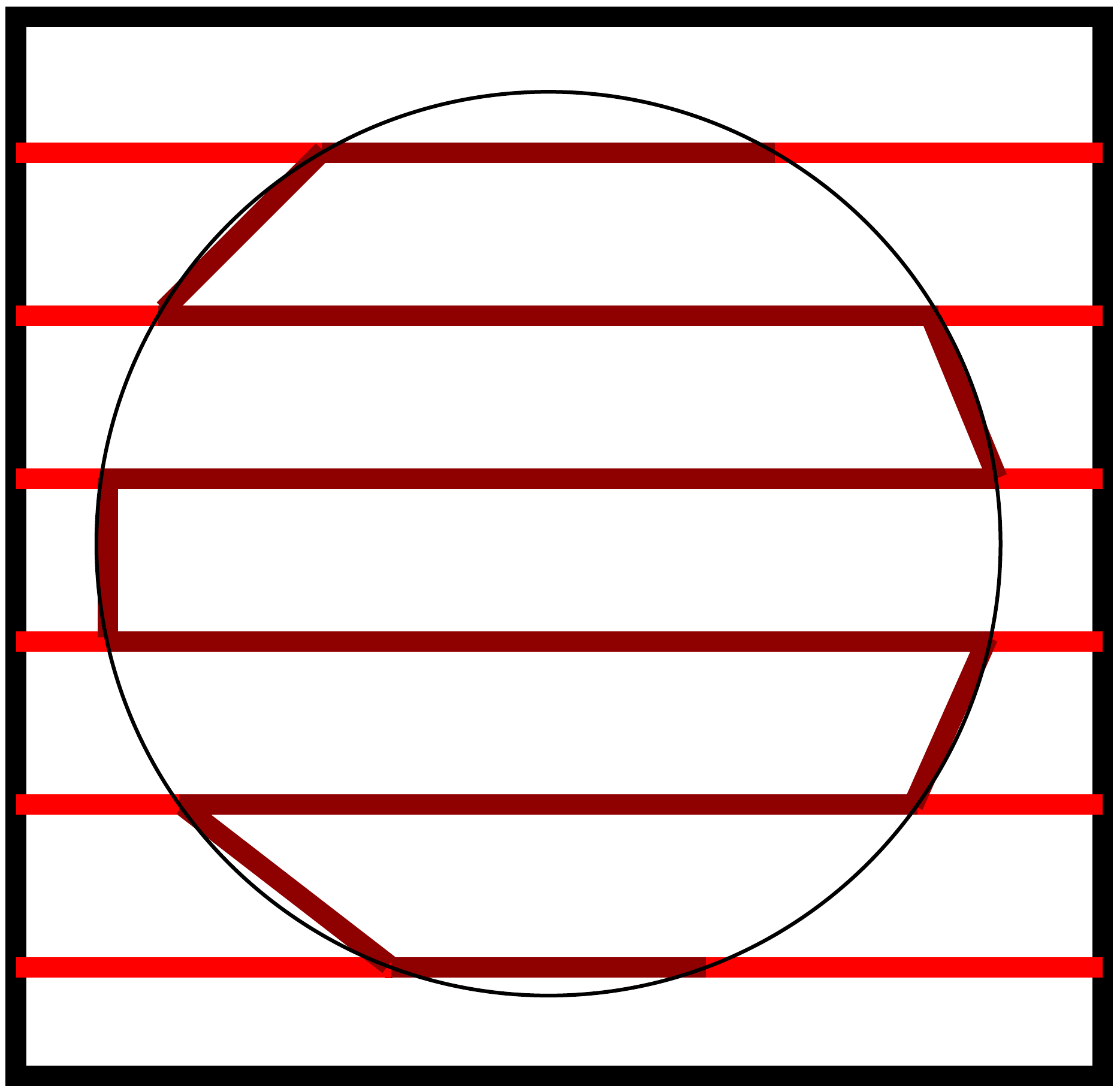}
\caption{An illustration of a curve (shown in dark) that satisfies condition~\Con{path3} for the trajectory set in
Figure \ref{fig:pd}.}
\label{fig:c2condn}
\end{figure}

Definition~\ref{trajset} is related to the concept of sampling
measures. A (positive Radon) measure $\mu $ on $\Re ^d $ is called a
sampling measure for $\mathcal{B}_\Omega$, if there exist $A,B>0$ such that
$$
A\|f\|_2^2 \leq  \int _{\Re ^d} |f(r)|^2 \, d\mu ( r)  \leq B\|f\|_2^2
\qquad \text{ for all } f \in \mathcal{B}_\Omega \, .
$$
If $\mu $ is a sum of point measures,  $\mu = \sum _{\lambda \in \Lambda } \delta _\lambda $, then one
recovers \eqref{eqn:frameboundsa}. For a trajectory set $P$ one can
define a natural measure, namely the sum of line integrals along the
trajectories $p_i$. Precisely, if each $p_i$ is parametrized by $\gamma _i
: I_i \to \Re ^d$ for some (finite or infinite) interval $I_i$, then
the corresponding path sampling measure is defined as
$$
\int _{\Re ^d} f(r) \,  d\mu (r) = \sum _{i\in \bI} \int _{I_i} f(\gamma
_i(t)) \gamma ' _i (t)\, dt
$$
(for $C^1$-curves $p_i$, otherwise we use the Riemann-Stieltjes
integrals $d\gamma _i(t)$).

An interesting line of research in complex analysis investigates
sampling measures and their characterizations for spaces of complex
functions, such as Fock space, Bergmann spaces, and also spaces of
bandlimited functions. See~\cite{Lue00,OC98} for a representative list of
contributions. Roughly speaking, if $\mu$ is a sampling measure, then
its support contains a set of sampling, although the
precise formulations are much more delicate and technical.

We prefer Definition~\ref{trajset} to the abstract definition
of sampling measures,  because our definition models
faithfully the acquisition of data  by mobile sampling: the samples
are taken along a path, possibly with high density, whence condition
~\Con{recon2}. The requirement of a realistic motion of the sensors  in  $\Re ^d$ leads
to the regularity condition ~\Con{path3}.

For a more quantitative version of stable trajectory sets we restrict
the range of the stability parameters $A,B$.
\begin{definition}
A trajectory set $P$ of the form (\ref{eqn:trajset}) is called a
\emph{stable Nyquist trajectory set for $\bom$ with stability parameters $A$ and $B$ } if $P$
satisfies condition~\Con{path3} and the modified condition
\begin{quote}
\nobreak
\begin{conditionAB}
\item \label{conAB:reconAB}\emph{[Nyquist]} There exists a uniformly discrete set $\Lambda$ of points on the
trajectories
in $P$, $\Lambda \subset P$,   such that $\Lambda$ forms a set of
stable sampling for $\bom$ with fixed stability parameters $A$ and $B$.
\end{conditionAB}
\end{quote}
\end{definition}
\noindent We denote  the collection of all stable Nyquist
trajectory sets  for $\Omega$ with stability parameters $A,B$ by
$\nomab$. Then $ \nom = \bigcup_{0 < A,B < \infty} \nomab.$

The sampling theory for mobile sensing is primarily concerned with
identifying suitable trajectory sets in $\nom$ and $\nomab$. The
key optimization problem is the  identification and description of
trajectory sets with minimal path density from these
classes:
\begin{equation}
\inf_{P \in \nom} \ell^{+}(P),\label{eqn:minpd}
\end{equation}
and
\begin{equation}
\inf_{P \in \nomab} \ell^{+}(P).\label{eqn:minpdAB}
\end{equation}

In \cite{unnvet13TIT} and \cite{unnvet12ISIT} we identified various examples of trajectory sets in $\nom$, and
obtained partial solutions to (\ref{eqn:minpd}) for restricted
classes of trajectories, for instance
uniform sets and unions of uniform sets. In this paper we derive a
lower bound of the path density for  the entire class of trajectories
consisting of arbitrary parallel lines, and we
study  the well-posedness of the optimization problem in both $\nom$ and $\nomab$.

\section{Optimal stable sampling sets composed of parallel lines}\label{sec:optresults}
In this section we consider a trajectory set composed of parallel lines in $\Re^d$. For these trajectories,
the path density coincides with the Beurling density of a cross-section.
\begin{lemma}
\label{lemma_dens_parallel}
Let $P$ be a trajectory set consisting  of lines parallel to a vector $q \in \Rdst \setminus\{0\}$ and let
$\Lambda := P \cap q^\perp$ be the intersection of $P$ with the hyperplane orthogonal to $q$. Then
$D^{-}(\Lambda)=\ell^{-}(P)$ and $D^{+}(\Lambda)=\ell^{+}(P)$.

In particular $P$ is homogeneous if and only if $D^{-}(\Lambda)=D^{+}(\Lambda)$
and in this case
\begin{align*}
\ell(P)=D^{-}(\Lambda)=D^{+}(\Lambda)=\lim_{a \longrightarrow +\infty} \frac{\#(\Lambda \cap
B_a^{d-1}(x))}{\abs{B^{d-1}_a}},
\mbox{ for all }x \in q^\perp \cong \mathbb{R}^{d-1}.
\end{align*}
\end{lemma}
\begin{proof}
The lemma is clear if in the definition of Beurling and path density we use cubes with sides aligned to $q^\perp$
instead of Euclidean balls. Lemma \ref{lemma_dens_tile} allows us to make this choice.
\end{proof}

Most practically useful parallel trajectory sets such as uniform sets, approximately uniform sets (e.g., with bounded
offsets) and their finite unions are homogeneous. To formalize the optimization problem we introduce some classes of
trajectories. For $\Omega \subset \Re^d$ we define:
\begin{itemize}
\item $\Par^q$: the class of all Nyquist trajectories consisting of lines parallel to $q$, ($q \in \Rdst
\setminus\{0\}$ a direction parameter).
\item $\Hom^q$: the class of all homogeneous Nyquist trajectories consisting of lines parallel to
$q$, ($q \in \Rdst \setminus\{0\}$ a direction parameter).
\item $\Par$: the union of the classes $\Par^q, q \not= 0$; that is, the collection of all trajectories consisting of
parallel lines.
\item $\Hom$: the union of the classes $\Hom^q, q \not= 0$; that is, the collection of all homogeneous trajectories
consisting of parallel lines.
\end{itemize}
The following is our main result about sampling along parallel lines.
\begin{thm}
\label{th:parallel}
Let $\Omega \subset \Re^d$ be a centered symmetric convex body. Then
\begin{align}
\label{eq:homq}
\inf_{P\in \Par^q} \ell^{-}(P) =
\inf_{P\in\Hom^q} \ell(P) = |\Omega \cap q^\perp|.
\end{align}
In particular, by optimizing over all $q \in \Re^d \setminus\{0\}$,  it follows that
\begin{align}
\inf_{P\in \Par} \ell^{-}(P) =
\inf_{P\in\Hom}\ell(P) = \min_{q \in\Re^d \setminus \{0\}} |\Omega \cap q^\perp|  .\label{eqn:inf}
\end{align}
\end{thm}
This  result shows that the lowest path density of a  set of parallel trajectories that admits stable
sampling of a field bandlimited to a convex, compact and symmetric set $\Omega$ is given by the volume of the smallest
section of $\Omega$ through the origin. 
Furthermore, this density can be almost attained by a homogeneous trajectory set.
This result is in the spirit of Landau's result \cite{lan67} on the minimum sampling density for stable
pointwise sampling, as illustrated in Figure \ref{fig:fundlimits}. In the rest of this section we present arguments that
build up to this result.

\begin{figure}
\centering
\subfigure[Classical pointwise sampling: Minimum sampling density $\propto \mbox{Vol}(\Omega)$.]{
\includegraphics[height=2.15 in]{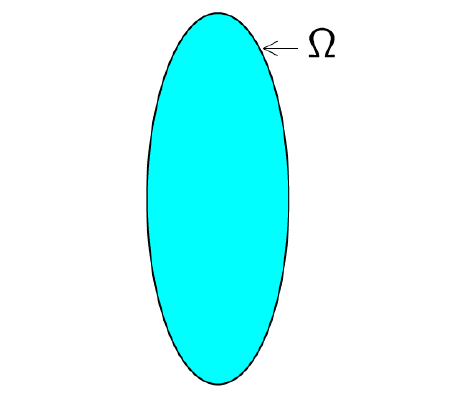}
\label{fig:omvol}
}
\subfigure[Sampling on parallel lines: Minimum path density $\propto$ Volume of minimum section through the center of
$\Omega$.]{
\centering
\includegraphics[height=2.15 in]{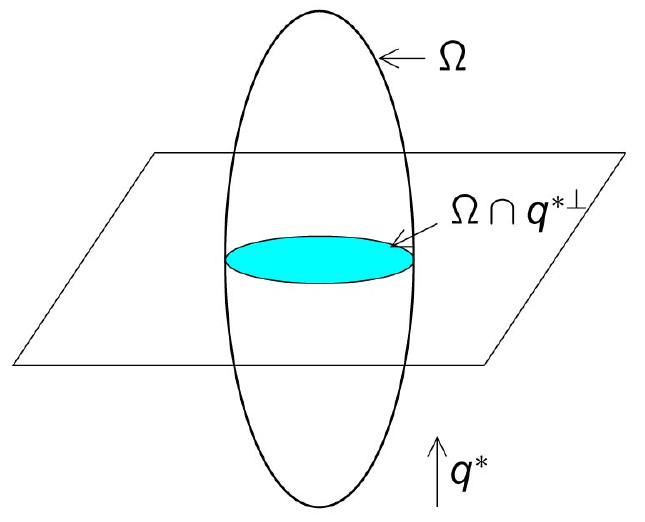}
\label{fig:sampR2c}
}
\caption{Fundamental sampling limits for a convex symmetric body $\Omega$.}
\label{fig:fundlimits}
\end{figure}

Sampling along parallel lines has been studied early on as an 
extension of non-uniform sampling theorems from 1-D to 2-D under the
name of \emph{line sampling}~\cite{BH89,GS01}. In particular, in
~\cite{GS01} sampling sets of the form $(x_i, y_{ik})$ are studied with
non-uniformly spaces lines at $x_i$ and non-uniformly spaced samples
$y_{ik}$ along each line. Let us emphasize that  our objective is
rather  different, as we try to  understand the  relation between the path
density and the sampling pattern consisting of parallel lines. This is
not about a particular set of parallel lines (as in the literature on
line sampling), but about all sets of parallel lines. Note that in
Theorem~\ref{th:parallel} we characterize  the optimal direction in which the
sensors have to move. This is an entirely  new aspect of the sampling problem.
\subsection{Regularity of paths of parallel lines}
The first step towards the proof of Theorem \ref{th:parallel} is showing that the trajectory sets
consisting of parallel lines based on a set with finite upper Beurling density do satisfy the regularity
condition \Con{path3}. To this end we need  the following lemma on the length of the shortest path that passes through
a given set of points (for a proof see \cite{beahalham59}).
\begin{lemma}
\label{lemma_short_path}
For every $d \geq 2$ there exists a constant $C_d>0$ with the following property: let $a>1$,
and $x_1, \ldots, x_n \in  B^d_a(0)$. Then there exists a continuous curve
$\alpha:[0,1] \to B^d_a(0)$ consisting of $n-1$ concatenated line segments that contains each point
$x_i$, $i=1, \ldots n$,  and
\begin{align*}
\length(\alpha) \leq C_d \,  a \,  n^{\frac{d-1}{d}}.
\end{align*}
\end{lemma}
We can now prove the following.
\begin{lemma}
\label{lem:c2forQ}
Let $P$ be a trajectory set consisting of lines parallel to a vector $q \in \Rdst \setminus\{0\}$ and let
$\Lambda := P \cap q^\perp$ be the intersection of $P$ with the hyperplane orthogonal to $q$.
Assume that $D^{+}(\Lambda) < \infty$. Then the trajectory set $P$ satisfies condition~\Con{path3}.
\end{lemma}
\begin{proof}
For every ball $B_a^d(x) \subset \rd $, we need to construct a single path
$\alpha$  containing all line segments of $P \cap B_a^d(x)$, but  without
increasing  the path length significantly. For this we need to connect
the points of intersection of $P \cap \partial B_a^d(x)$  on each
hemisphere by a short path. 
Such a choice in dimension $d=2$ is plotted  in
Figure~\ref{fig:ext}. In higher dimensions, we resort to
Lemma~\ref{lemma_short_path}. 

For a rigorous argument, we may 
 assume without loss of generality that the lines in $P$ are parallel to $e_d= (0,\ldots,0,1)$.
Let $x \in \Re^d$ be arbitrary.
Let $H^+$ and $H^-$ denote the half-spaces
\[
H^+ = \{y \in \Re^d: y_d > x_d\} \quad \mbox{ and } \quad H^- = \{y
\in \Re^d: y_d < x_d\} \, ,
\]
and $H$  the hyperplane
\[
H = \{y \in \Re^d: y_d = x_d\}.
\]
Let $A^+ := P \cap \partial B_a^d(x) \cap H^+$,$A^- := P \cap \partial B_a^d(x) \cap H^-$
and $A := P \cap \partial B_a^d(x) \cap H$.
To each point $y=(y_1,y_2,\ldots,y_d) \in A^+$ corresponds a symmetric
point $y^- =  (y_1,y_2,\ldots,y_{d-1},2x_d-y_d) \in A^-$.
Let us further denote $N := \#A^+=\#A^-$ and $M:=\# A$. Since $D^+(\Lambda) < +\infty$,
it follows that
\begin{align}
\label{eq_ns}
N, M =O(a^{d-1}).
\end{align}
By Lemma \ref{lemma_short_path}, there exists a path $\SP$ contained in $B_a^d(x)$ consisting of $N-1$ line segments,
that passes through all the points in $A^+$, and such that
$\length(\SP) = O(a) N^{\frac{d-1}{d}} \asymp a a^{\frac{(d-1)^2}{d}}=a^{d+\tfrac{1}{d}-1}=o(a^d)$.

Let the sequence $a_1,a_2,\ldots,a_N$ denote the order in which points in $A^+$ appear in $\SP$.
By symmetry, the sequence of line segments connecting the points $a_1^-,a_2^-,\ldots, a_N^-$ is a path contained
in $B_a^d(x)$ that connects all points in $A^-$ and has length $O(a) N^{\frac{d-1}{d}}=o(a^d)$.

\begin{figure}
\centering
\includegraphics[scale=0.3]{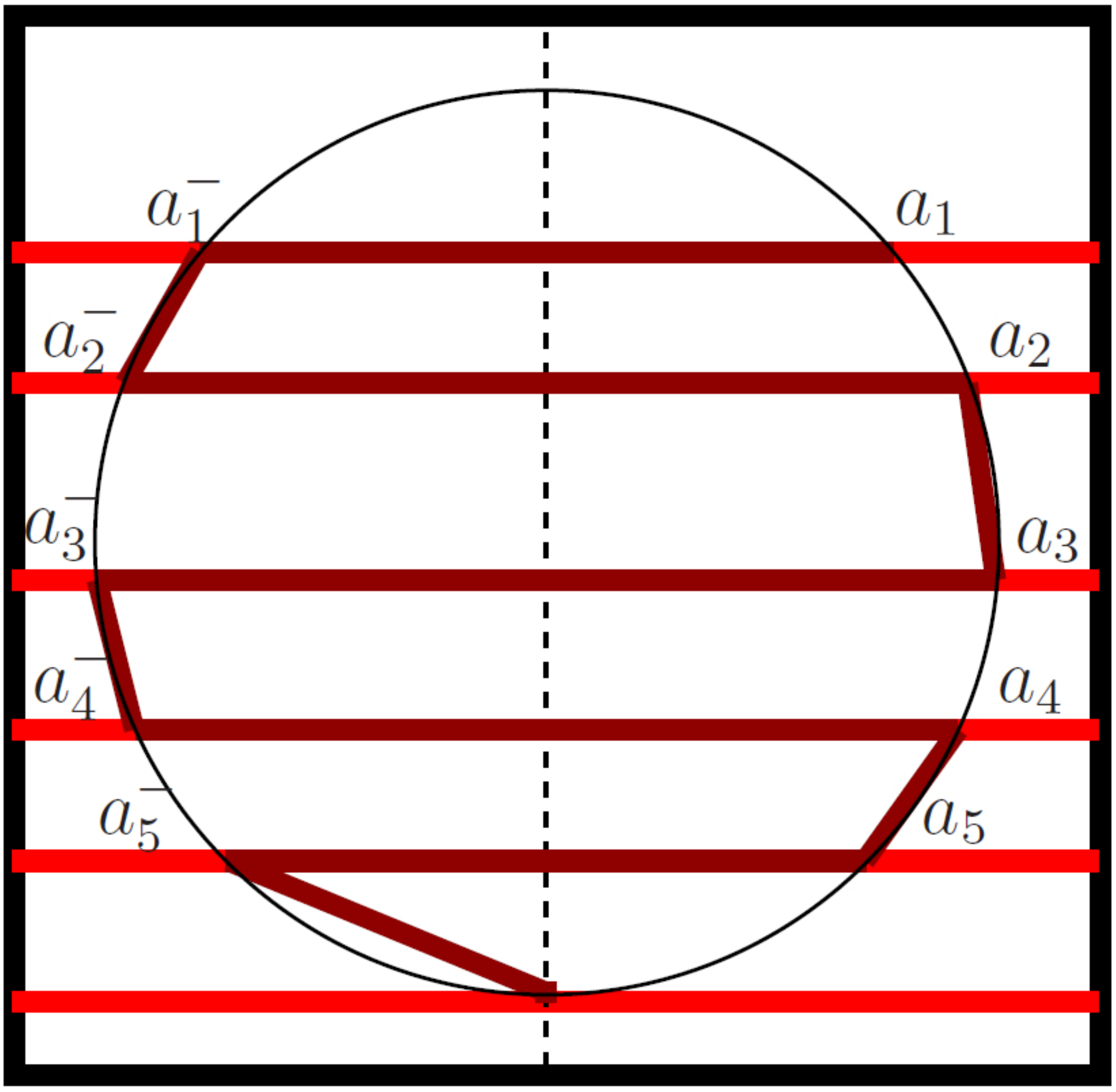}
\caption{The intersection of $P$ and ball is prolongated to a rectifiable path with negligible additional length. }
\label{fig:ext}
\end{figure}
We construct a rectifiable curve $\alpha$ containing $P \cap B_a^d(x)$ as follows.
Let $\beta$ denote the curve comprising the sequence of line segments connecting the points
\begin{align*}
&a_1,a_1^-,a_2^-,a_2,a_3,a_3^-,a_4^-, a_4,a_5,\ldots, a_N^-,a_N,
\quad \mbox{if $N$ is even},
\\
&a_1,a_1^-,a_2^-,a_2,a_3,a_3^-,a_4^-, a_4,a_5, \ldots, a_N,a_N^-,
\quad \mbox{if $N$ is odd}.
\end{align*}
Since for all $1\leq i \leq N$, the curve $\beta$ contains the line segment connecting $a_i$ and $a_i^-$ exactly once,
it follows that $\beta$ contains $(P \cap B_a^d(x)) \setminus H$.
Furthermore for all $1\leq i \leq N-1$, the curve $\beta$ contains either the line segment connecting $a_i$ and
$a_{i+1}$ or that connecting $a_i^-$ and $a_{i+1}^-$.
Thus counting all line segments in $\beta$ we obtain
\begin{align}
\label{eq_beta}
\length(\beta) = \length(\SP) + \clD^P(a,x).
\end{align}
Invoking again Lemma \ref{lemma_short_path} we obtain a curve $\beta'$ contained in $B^d_a(x)$ that goes through each
point in point $A$ and has length
$O(a) M^{\frac{d-1}{d}} \asymp a^{d+\tfrac{1}{d}-1}=o(a^d)$.
Finally we form $\alpha$ by linking $\beta'$ to $\beta$ by means of a line segment
contained in $B^d_a(x)$ (of length at most $a$).

The curve $\alpha$ is completely contained in $B^d_a(x)$, it contains $P \cap B_a^d(x)$
and it is rectifiable since it consists of a finite number of line segments. In addition, from the length estimates
above we conclude that $\length(\alpha) = \clD^P(a,x) + o(a^d)$, as desired. (Note that in all the estimates, the
implicit constants depend on the set $\Lambda$ but not on the center of the ball $x$.)
\end{proof}
\begin{remark}
{\rm
For the proof of  Lemma \ref{lem:c2forQ} we do not need the full
strength of  Lemma~\ref{lemma_short_path}. If we accept (without proof) that
in condition~\Con{path3} instead of balls $B_a^d(x)$ one may use cubes $Q_a$ with side length $2a$
 and aligned  parallel to lines in $P$,
 then Lemma~\ref{lemma_short_path} can be replaced by the following,  more  elementary
 argument. If a cube $Q_a$ is parallel to $P$, then
 $P\cap \partial Q_a$ contains two copies of $P\cap q^\perp$. As $\Lambda = P \cap q^\perp $ is relatively separated, it
 can be approximated by a finite union of lattices isomorphic to
 $\mathbb{Z}^{d-1}$ (with asymptotically small error). It is now
 elementary to connect the lattice points in $\mathbb{Z}^{d-1} \cap
 [-a,a]^{d-1}$ by a path of length at most $(2a)^{d-1}$. The proof of
 Lemma~\ref{lem:c2forQ} remains unchanged.}
\end{remark}

\subsection{Lower bounds for the path density}
\begin{proposition}\label{prop:unifbetterthanparallel}
Let $\Omega \subset \Re^d$ be a convex centered symmetric body. Let $P \in \Par$ be a
Nyquist trajectory set composed of lines parallel to $q \in \Re^d \setminus \{0\}$.
Then $\ell^{-}(P) \geq |\Omega \cap q^\perp|$.
\end{proposition}
\begin{proof}
After a rotation, we may assume without loss of generality that $q =
e_d=(0,0, \dots, 0,1)$.
Denote $\Omega \cap q^\perp = \Omega_0 \times \{0\}$ with $\Omega_0 \subseteq \Rst^{d-1}$.
Let $\delta\in(0,1)$ and consider the set $(1-\delta) \Omega$. By Lemma \ref{lemma_conv_1},
$(1-\delta) \Omega \subseteq \Omega^\circ$. Hence $(1-\delta) \Omega$ and $\partial \Omega$ are two disjoint
compact sets and consequently
\begin{align*}
\varepsilon := d((1-\delta) \Omega, \partial\Omega) >0.
\end{align*}
This implies that
\begin{align*}
(1-\delta) \Omega_0 \subseteq
\Omega_0^\varepsilon := \set{x \in \Omega_0}{d((x,0), \partial\Omega)
  \geq \varepsilon} \subseteq \Re ^{d-1} ,  
\end{align*}
where $d((x,0), \partial\Omega)$ denotes the Euclidean distance from $(x,0)$ to the set $\partial\Omega$.
Let $\Lambda \subseteq \Rst^{d-1}$ be the set at which the lines in
$P$ intersect the hyperplane $e_d^\perp$, i.e.,
$P = \sett{(\lambda,t): t \in \Rst}= \Lambda \times \Re $.
Since $P$ is a Nyquist trajectory set, assumption~\Con{recon2} implies the
existence of  a sampling set  $\Gamma \subseteq \Rdst$ for
$\bom$ whose points belong to the trajectories in $P$. Hence $\Gamma \subseteq \Lambda \times \Rst$.
For each $\lambda \in \Lambda$, let $I_\lambda := \set{t \in
  \Rst}{(\lambda,t) \in \Gamma}$.

Let $\nu_d(\Gamma) = \max _{x\in \rd } \# (\Gamma \cap  Q_{1/2}(x) )$ be
the covering constant of $\Gamma \subset \Re ^d$. Since $\Gamma $ is
a set of stable sampling for $\bom$, its upper Beurling density $D^{+}(\Gamma)$ is finite and
consequently $\nu_d(\Gamma) < +\infty$.
Hence, for all $\lambda \in \Lambda$,
\begin{align*}
\nu_1(I_\lambda) = \nu_d(\{\lambda\} \times I_\lambda) \leq
\nu_d(\Gamma)  < \infty \, .
\end{align*}
Let $g \in L^2(\Rst^{d-1})$ be bandlimited on  $\Omega_0^\varepsilon$ and set $f(x) =
g(x_1,x_2,\ldots,x_{d-1})\text{sinc}(\varepsilon x_d)$ with
$\text{sinc} (x) = \tfrac{\sin (\pi x)}{\pi x}$ as usual. Since $\Omega
_0^\epsilon \times [-\epsilon, \epsilon ]\subset \Omega $, we have  $f \in \bom$. Using the fact that $\Gamma$ is a
sampling
set for $\bom$ we have
\begin{align*}
\varepsilon^{-1} \norm{g}_2^2 = \norm{f}_2^2
&\lesssim \sum_{\gamma \in \Gamma} \abs{f(\gamma)}^2
=
\sum_{\lambda \in \Lambda} \abs{g(\lambda)}^2 \sum_{t \in I_\lambda} \abs{\text{sinc}(\varepsilon t)}^2
\\
&\lesssim \varepsilon^{-1} \nu_d(\Gamma ) \sum_{\lambda \in \Lambda} \abs{g(\lambda )}^2.
\end{align*}
Hence $\norm{g}_2^2 \lesssim \sum_{\lambda \in \Lambda} \abs{g(\lambda)}^2$ for every $g$ bandlimited to
$\Omega_0$. By Landau's result on necessary density conditions for sampling \cite{lan67} we deduce that
$D^{-}(\Lambda) \geq |\Omega_0^\varepsilon|$. Thus, by Lemma \ref{lemma_dens_parallel}, it follows that
\begin{align*}
\ell^{-}(P)=D^{-}(\Lambda) \geq |\Omega_0^\varepsilon| \geq
\abs{(1-\delta)\Omega_0}
=(1-\delta)^{d-1} \abs{\Omega_0}.
\end{align*}
 The conclusion follows because  $\delta >0$ was arbitrary and
 $|\Omega _0| = |\Omega \cap q^\perp |$.
\end{proof}

\subsection{Reduction to sampling in each section}
To prove that equality holds in \eqref{eq:homq} we must show that there are Nyquist trajectories with path density
arbitrarily close to the volume of the section of $\Omega$ through the origin. The following proposition shows that
this problem can be reduced to finding sampling sets for each section
of $\Omega $  with uniform bounds. Precisely, for $t\in \Re $ let
\begin{align*}
\Omega _t = \{ x \in \Re ^{d-1} : (x,t) \in \Omega \}  \subset \Re^{d-1}
\end{align*}
be the section of $\Omega $ at height $t$. Then $\Omega = \bigcup
_{t\in \Re } \big(\Omega _t \times \{t\}\big)$.

\begin{proposition}
\label{prop_sections}
Let $\Omega \subseteq \Rdst$ be a closed set.
Assume that $\Lambda=\sett{\lambda_k: k \geq 1} \subset \Re ^{d-1} $
is a set of stable sampling  for $\mathcal{B}_{\Omega_t} \subseteq L^2(\Re ^{d-1})$ with uniform bounds
$0<A \leq B< \infty$ for all $t \in \Rst$,  then for every $f \in {L^2(\Rdst)}$ with
$\supp(\widehat{f}) \subseteq \Omega$
\begin{align}
\norm{f}_2^2 \asymp \sum_{k \geq 1} \int_{\Rst} \abs{f(\lambda_k,t)}^2 dt.
\end{align}
If in addition $\Omega$ is compact, then there exists a lattice
$\Gamma \subseteq  \Re $,
such that
\begin{align}
\norm{f}_2^2 \asymp \sum_{k \geq 1} \sum_{\gamma \in \Gamma}
\abs{f(\lambda_k,\gamma)}^2 \, .
\end{align}
\end{proposition}
\begin{proof}
Set   $g = \widehat{f}$ and $g_t(x')= g(x',t)$ for  $x=(x',t) \in
\Rst^{d-1}\times\Rst$ . Then $\supp(g) \subseteq \Omega$ and $\supp
(g_t) \subset \Omega _t$. We further define the partial Fourier
transform
\begin{align*}
G_k(t):= \int_{\Rst^{d-1}} g(x',t) e^{2\pi {\sf i} \langle \lambda_k,
  x'\rangle} \, dx' = \hat{g_t}(\lambda _k) .
\end{align*}

Since
\begin{align*}
\int_{\Rst} G_k(t) e^{2\pi {\sf i} w t} \,dt =
\int_{\Rdst} g(x',t) e^{2\pi {\sf i} \langle \lambda_k, x'\rangle}
e^{2\pi {\sf i} w  t } \, dx' \, dt=f(\lambda_k,w) \, ,
\end{align*}
Plancherel's theorem yields
\begin{align*}
\int_\Rst \abs{G_k(t)}^2 dt =\int_\Rst \abs{f(\lambda_k,w)}^2 \, dw.
\end{align*}
Using the support property $\supp(g_t) \subseteq \Omega_t$, we obtain
for almost all $t\in \Re $ that
\begin{align*}
\sum_k \abs{G_k(t)}^2 &= \sum_k \abs{\widehat{g_t}(-\lambda_k)}^2
\asymp \int_{\Rst^{d-1}} \abs{g_t(x')}^2 \,dx'
\\
&= \int_{\Rst^{d-1}} \abs{g(x',t)}^2 \,dx' \, ,
\end{align*}
with constants independent of $t$ by assumption. Finally,
\begin{align*}
\sum_{k \geq 1} \int_{\Rst^{d-1}} \abs{f(\lambda_k,t)}^2 dt
&= \sum_{k \geq 1} \int_\Rst \abs{G_k(t)}^2 dt
= \int_\Rst \sum_{k \geq 1} \abs{G_k(t)}^2 dt
\\
&\asymp
\int_\Rst \int_{\Rst^{d-1}} \abs{g(x',t)}^2 \,dx' dt
= \norm{g}_2^2 = \norm{f}_2^2.
\end{align*}
If in addition $\Omega$ is compact, then $f$ is bandlimited to a compact set
and the integrals involving $f$ can be replaced by sums over a
suitably dense lattice.
\end{proof}
\begin{remark} {\rm 
Proposition \ref{prop_sections} applies to spectra of the form
\begin{align*}
\Omega = \sett{(x',t)\in \Rst\times\Rst: g_1(x') \leq t \leq g_2(x')}
\end{align*}
with two continuous functions $g_1, g_2$. 
This set can have a very large projection onto the last coordinate
while $\abs{g_2(x')-g_1(x')}$ remains small. }
\end{remark}
\subsection{Universal sampling sets}
In order to prove Theorem \ref{th:parallel},  we need to find, for
each centered symmetric convex body $\Omega$ and each 
direction $q \not= 0$, a stable Nyquist trajectory for $\bom$
consisting of lines parallel to $q$ and with a  path-density close to
the measure of the central section of $\Omega$ by $q^\perp$. After a
rotation, we may assume that $q = e_d$ and analyze the  horizontal sections $\Omega_t := \sett{(x_1, \ldots, x_{d-1}):
(x_1, \ldots, x_{d-1}, t) \in \Omega}\subseteq \mathbb{R}^{d-1}$ of
$\Omega $. According to  Proposition \ref{prop_sections}, we need to  find
a set $\Lambda \subseteq \Rst^{d-1}$, such that (a) its   Beurling
density is close to $\abs{\Omega_0}$ and (b) $\Lambda $ is
simultaneously a sampling set for all  spaces of functions bandlimited
$ \mathcal{B}_{\Omega_t}$ for all $t$ with uniform sampling bounds.

In the special  case when  $\Omega$ is contained in an ``oblique''  cylinder, i.e.,
$\Omega _t \subseteq  tv + \Omega _0$ for some vector $v\in
\mathbb{R}^{d-1}$ and  all  $t$ (Figure~\ref{fig:sampR2c}),  it suffices to find a  sampling set only for
 $\mathcal{B}_{\Omega _0}$  with density close to the  critical
 one. This problem was already  solved  in \cite{marzo}.

 In general, the horizontal sections $\Omega _t$ are \emph{not}
 contained in translates of  the central section $\Omega _0$. As a simple example we
 mention the regular octahedron and two sections perpendicular to
 $(1,1,1)$.  The octahedron fits into a cylinder with a cross-section
 that is strictly  larger than the central \emph{minimal} cross-section (see Figure
\ref{fig_secs}). Therefore  the
 simple argument sketched above 
 does not work. To solve the general case, we need the concept of
 universal sampling sets, as introduced in~\cite{olul08,
mame10}.
\begin{figure}
\centering
\subfigure[The central section by the plane $x+y+z=0$ (shaded hexagon).]{
\includegraphics[height=2.15 in]{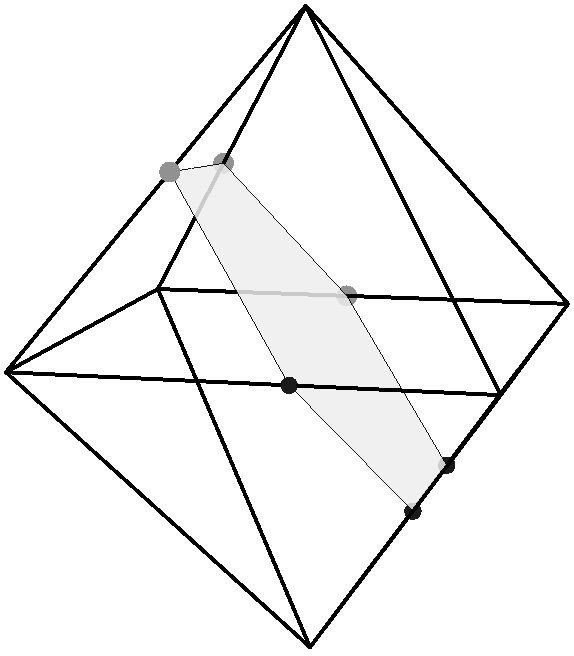}
}
\subfigure[The parallel section by the plane $x+y+z=1$ is a triangle
(top face of
the octahedron) is not contained in the central one (shaded hexagon).]{
\centering
\includegraphics[height=2.15 in]{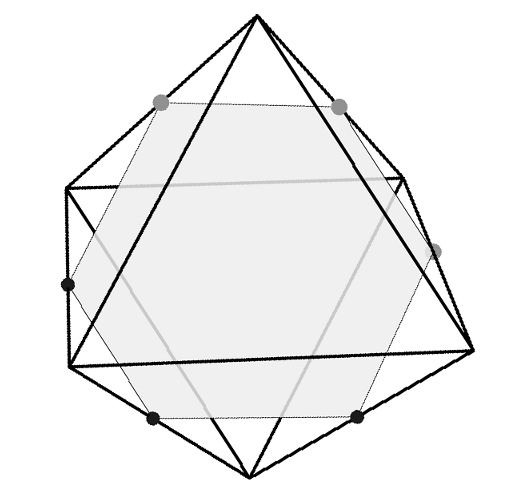}
}
\caption{The regular octahedron $\sett{(x,y,z) \in \Rst^3: \abs{x}+\abs{y}+\abs{z} \leq 1}$ and two parallel sections.}
\label{fig_secs}
\end{figure}

Given $\eta>0$, a \emph{$\eta $-universal sampling set} $\Lambda$ is a set
with uniform density $\eta$ that is a sampling set for $\bom$, for all compact spectra $\Omega \subseteq \Rdst$
with $\abs{\Omega} < \eta$. It is known that for all $\eta>0$ there
exist  universal sampling sets \cite{olul08,
mame10}. For example, in dimension $d=1$ the set $\{ n + \{\sqrt{2}n\} \} : n\in \mathbb{Z}\}$ is a universal sampling
set with density $\eta
=1$ (with $\{x\} = x - \lfloor x \rfloor$ denoting
the fractional part of $x$). 
 On the other hand, if the requirement that $\Omega$ be compact is dropped, universal sampling sets do not exist
\cite{olul08}.

A universal $\eta$-sampling set is a  set of stable sampling  for all compact spectra $\Omega$ with $\abs{\Omega} <
\eta$,
but the frame bounds may depend on $\Omega$. We now argue that when the spectra consist of sections of a compact
convex body, then these bounds can be chosen to be uniform. We need the following technical lemma,
whose proof is deferred to  Section \ref{sec_proof_lemma_sections}.

\begin{lemma}[Continuity of the sections]
\label{lemma_sections}
Let $\Omega \subseteq \Rdst$ be a convex and compact set, $t \in \Rst$, and $\varepsilon>0$. Then there exists
$\delta>0$ such that for all $s \in (t-\delta,t+\delta)$
\begin{align*}
\Omega_{s} \subseteq \Omega_t + B_\epsilon  ^{d-1} \, .
\end{align*}
\end{lemma}
We now show that the sections of a convex compact set admit a universal sampling set with uniform stability
bounds.

\begin{proposition}
\label{prop_unif_bounds}
Let $\Omega \subseteq \Rdst$ be a convex and compact set and let
\begin{align*}
\eta > \max_{t\in\Rst}{\abs{\Omega_t}}.
\end{align*}
Let $\Lambda$ be an $\eta$-universal sampling set. Then $\Lambda$ is a sampling set for
all $\Omega_t$, $t \in \Rst$, with sampling bounds uniform in $t$.
\end{proposition}
\begin{proof}
Let $I \subseteq \Rst$ be compact interval such that
$\Omega \subseteq \Rst^{d-1} \times I$. Let $t \in I$. Since $\Omega_t$ is closed,
there exists $\varepsilon_t>0$ such that
\begin{align*}
\abs{\Omega_t + B_{\varepsilon_t}} < \eta.
\end{align*}
We let $\widetilde{\Omega}_t := \Omega_t + B_{\varepsilon_t}$ denote the slightly enlarged section.

With this notation, by Lemma \ref{lemma_sections}, there exists $\delta_t>0$ such that
\begin{align}
\Omega_s \subseteq \widetilde{\Omega}_t,
\mbox{ if } s \in (t-\delta_t,t+\delta_t).
\end{align}
 The family of intervals $\{(t-\delta_t,t+\delta_t): t \in I\}$
is an open cover of $I$. Then, by compactness,  $I \subseteq \bigcup
_{k=1}^N (t_k-\delta_{t_k},t_k+\delta_{t_k}) $  for finitely many
$t_k \in \Re $.
Hence, for every $s \in I$, there exists $k \in \{1, \ldots, N\}$ such that
\begin{align}
\label{eq_inclusion}
\Omega_s \subseteq \widetilde{\Omega}_{t_k}.
\end{align}
Since $\abs{\widetilde{\Omega}_{t_k}} < \eta$ for $k=1,\dots, N$,
the universal sampling property implies that $\Lambda$ is a sampling set for
$\mathcal{B}_{\widetilde{\Omega}_{t_k}}$ with bounds $0<A_k \leq B_k < \infty$. Let
\begin{align*}
A &:= \min\{A_1, \ldots, A_N\},
\\
B &:= \max\{B_1, \ldots, B_N\}.
\end{align*}
Hence, $\Lambda$ is a sampling set for $\mathcal{B}_{\widetilde{\Omega}_{t_k}}$ with bounds $A,B$
for all $k=1, \ldots, N$. Since, according to \eqref{eq_inclusion}, every section
$\Omega_s$ is contained in some set $\widetilde{\Omega}_{t_k}$, it follows that
$\Lambda$ is a sampling set with bounds $A,B$ for all $\mathcal{B}_{\Omega_s}$ with $s \in I$.
Note finally that $\Omega_s=\emptyset$, for $s \notin I$. This completes the proof.
\end{proof}
\subsection{Upper path density bounds}
With Proposition \ref{prop_unif_bounds} we can now show the estimates
 \eqref{eq:homq} for the necessary path density for convex spectra.
\begin{proof}[Proof of Theorem \ref{th:parallel}]
From Proposition \ref{prop:unifbetterthanparallel} it follows that
\begin{align*}
\inf_{P\in\Hom^q} \ell(P) \geq
\inf_{P\in \Par^q} \ell^{-}(P) \geq |\Omega \cap q^\perp|.
\end{align*}
Let us show that all these inequalities are actually equalities.
Assume without loss of generality that $q=e_d= (0,\ldots,0,1)$ and note that since
$\Omega$ is convex and symmetric the section through the origin is the
one with maximal area. This is a consequence of the Brunn-Minkowski inequality, see for example \cite{ga02}.
Given a number $\eta$ satisfying
\begin{align*}
\eta > |\Omega \cap q^\perp|,
\end{align*}
let $\Lambda \subseteq \Rst^{d-1}$ be a $\eta$-universal sampling set and let
$P$ be a set of lines parallel to $q$ that go through $\Lambda$.
Since $\Lambda$ possesses finite (uniform) density,  $P$ satisfies
condition~\Con{path3} by Lemma \ref{lem:c2forQ}.
In addition, the fact that $\Lambda$ possesses a uniform density and Lemma
\ref{lemma_dens_parallel} imply that
$\ell(P)=D(\Lambda ) = \eta$ and that $P$ is homogeneous.  Propositions \ref{prop_unif_bounds} and
\ref{prop_sections} imply that $P$ is a Nyquist trajectory set. This shows that
$\inf_{P\in\Hom^q} \ell(P) \leq \eta$. The conclusion follows by letting $\eta$ tend to $|\Omega \cap q^\perp|$.
\end{proof}

\section{Optimizing over arbitrary trajectory sets}\label{sec:disc}
We now consider the problem of designing trajectory sets without
requiring the trajectories to be straight lines.
\subsection{Ill-posedness of the  unconstrained problem}
In the following proposition we show that the optimization problem  \eqref{eqn:minpd} is ill-posed by constructing a
sequence of trajectory sets in $\nom$  with arbitrarily small  path
density.
\begin{proposition}
\label{prop:infzero}
Let $\Omega \subseteq \Rst^2$ be a compact set. For every $\epsilon
>0$  there exists a trajectory set $P\in \nom$, such that $\ell ^+ (P)
<\epsilon $.  Thus,
\begin{equation*}
\inf_{P \in \nom} \ell^+(P) = 0.
\end{equation*}
\end{proposition}
\begin{proof}
By enlarging $\Omega$ if necessary, we can assume that it is a cube. Since the statement to be proved
is invariant under dilations we further assume that $\Omega = [-1/2,1/2]^2$. For each $n \geq 1$ we construct a
trajectory set $P_n$, in such a way that $\ell^+(P_n) \longrightarrow 0$, as $n \longrightarrow \infty$.

The counterexample is given by the path $P_n$ resulting from the set
\begin{align}
\label{eq_lines}
\Big(n \mathbb{Z} \times \Re \Big)  \cup \Big((n\mathbb{Z} + [0,\tfrac{1}{n}]) \times
\mathbb{Z} \Big) \, ,
\end{align}
which is the the union of vertical lines with spacing $n$
and small horizontal segments emerging at the point $(nj,k), j,k\in
\mathbb{Z}$.  See Figure \ref{fig:hairs}.
\begin{figure}
\centering
\includegraphics[scale=0.5]
{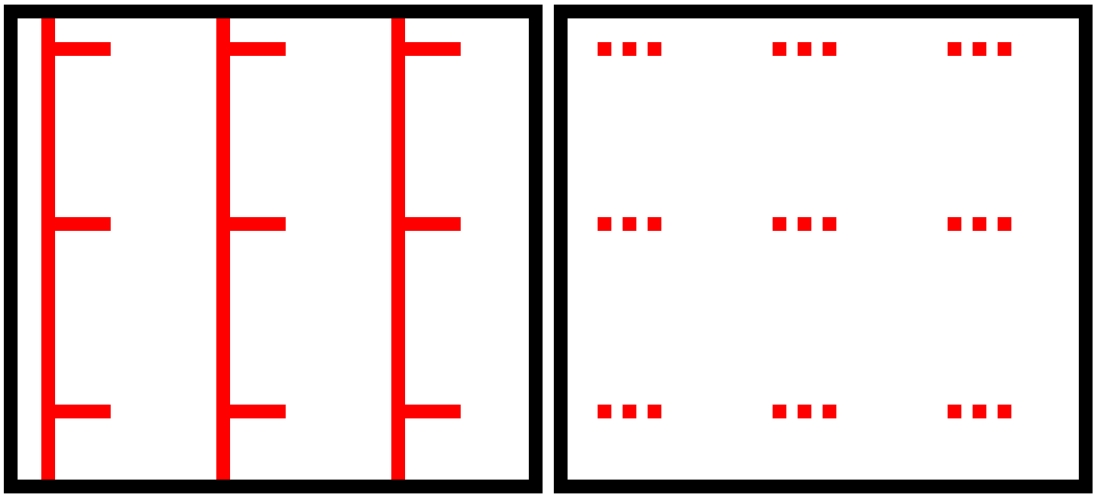}
\caption{Left: the path trajectory $P_n$.  Right: a set of stable sampling contained in the trajectory set. }
\label{fig:hairs}
\end{figure}
This construction ensures that $\ell^+(P_n) \lesssim 1/n$. Clearly
$P_n$ satisfies condition \Con{path3}. It remains to show that $P_n$ contains a sampling set
for $\mathcal{B}_{\Omega}$.

Let $F_n \subseteq [0,1/n]$ be a finite set of cardinality $2n$ and
$\Gamma_n := \{nk+t: k \in \Zst, t \in F_n\}$ its periodization with
period $n$. Then $\Lambda_n=\Gamma_n\times \Zst$ is separated and
contained in $P_n$. Since $D^{-}(\Gamma_n) = 2 >1$,  it follows that $\Gamma_n$ is a sampling
set for $\mathcal{B}_{[-1/2,1/2]}$, and consequently
$\Lambda_n$ is a sampling set for $\mathcal{B}_{\Omega}$.
\end{proof}

\begin{remark}
{\rm A similar example can be constructed in dimension $d$.}
\end{remark}

\begin{remark} {\rm 
The path density of a Nyquist trajectory $P \in \nom $  is always strictly positive,
thus the infimum in 
Proposition \ref{prop:infzero} is never attained. To see this, choose
a uniformly discrete subset $\Lambda \subseteq P$ that is a set of
sampling for $\bom $ (by condition   \Con{recon2}).  Let 
$\delta := \inf\sett{\abs{\lambda-\lambda'}: \lambda,\lambda' \in \Lambda,
\lambda \not= \lambda'} >0$ be the separation of $\Lambda $. 
 Since $\Lambda $
is a set of sampling for $\bom $, Landau's density result asserts that
$D^-(\Lambda ) \geq |\Omega |$~\cite{lan67}. This means that,  for fixed
$\eta , 0< \eta < |\Omega |$, and  sufficiently large $a >0$,  we have 
\begin{align*}
 \# (\Lambda \cap B^d_a(x)) \geq \eta \abs{B^d_a(x)} \quad \forall
 x\in \rd .
\end{align*}
For  $a $ sufficiently large  and $x \in \Rdst$,  let $\alpha: [0,1] \to \Rdst$ 
be the curve granted by condition~\Con{path3} in Definition \ref{trajset}.
Then  $\alpha ([0,1])  \supset P \cap B^d_a(x) \supset \Lambda \cap
B^d_a(x)$. Since the minimum distance between  points in $\Lambda$  is at least
$\delta$, it follows that 
\begin{align*}
\clD^P(a,x) + \littleo(a^d)
= \length(\alpha) \geq \delta N_\Lambda(a,x)
\geq \delta \eta \abs{B^d_a(x)}.
\end{align*}
Hence, }
\begin{align*}
\ell ^{-} (P)=\liminf_{a\longrightarrow \infty}
\inf_{x \in \Rdst}
\frac{\clD^P(a,x)}{\abs{B^d_a(x)}} \geq \delta \eta >0.
\end{align*}
\end{remark}
We conclude from Proposition \ref{prop:infzero} 
that the optimization problem  \eqref{eqn:minpd} which was first posed in \cite{unnvet13TIT} has a trivial solution. In
other words, for every compact set $\Omega$ it is possible to design a stable Nyquist trajectory set
for $\bom$ with arbitrarily small path density. Although at first glance this result may look counter-intuitive, a
closer
look at the sequence of trajectory sets in the counter-example reveals
that the condition number  $\frac{B}{A}$ of the sampling set from
\eqref{eqn:frameboundsa} diverges to $\infty$. Thus although we have a stable trajectory set, the stability margin may
be arbitrarily bad.

\subsection{Trajectory sets with given stability parameters}
One way to address the ill-posedness of this problem is to restrict the optimization to trajectory sets that contain
stable sampling sets with given stability parameters $A$ and $B$.
 In this section we show that this problem is indeed well-posed by identifying a
non-zero lower bound on the path density for every trajectory set in $\nomab$.

In order to obtain a lower bound on the path density we exploit the
key fact that the size of the largest hole of a sampling set is
determined by the condition number $B/A$ \cite{iosped00}.

\begin{proposition}
\label{prop:gaps}
\mbox{}

\begin{itemize}
\item[(a)] Let $\Omega \subseteq \Rdst$ be a compact set with a smooth
  boundary and surface measure $\sigma (\partial \Omega )$.  Let $\Lambda \subset \Rdst$ be a
sampling set
for $\mathcal{B}_\Omega$ with stability bounds $A,B$. Then $\Lambda$ intersects every
cube $x+ [-R,R]^d$,
where
\begin{align*}
R = C \frac{B \sigma (\partial\Omega)}{A \abs{\Omega}},
\end{align*}
and $C$ is a constant that depends only on $d$.
\item[(b)] If $\Omega=[-1/2,1/2]^d$, then $R$ may chosen explicitly as
\begin{align*}
R=\bound.
\end{align*}
\end{itemize}
\end{proposition}
\begin{proof}
Part (a) is a simplified version of the main result in
\cite{iosped00}. The result in~\cite{iosped00} is more general and  also covers the case of spectra
with fractal boundaries.

Part (b) follows from explicit estimates. In Section \ref{app:gap}
we give a full argument based on \cite{niol12}.
\end{proof}

For a measurable set $E \subseteq \Rdst$ we define a quantity $\Delta
_E $ by
\begin{align}
\label{eq_deff_deltae}
\Delta_E := \sup_{q \in \Rdst \setminus \{0\}} \abs{\clP_{q^\perp } E} .
\end{align}
This quantity is the volume in $\Re ^{d-1}$ of the maximal
projection  of $E$ onto a hyperplane. It satisfies the following 
invariance  properties:
\begin{align}
  \Delta_{E+x}& =\Delta_E     & \qquad \forall x \in \Rdst \, ,\label{slide_c1} \\
\Delta_{(1+\delta)E}& =(1+\delta)^{d-1}\Delta_E  & \qquad \forall \delta>0\,
. \label{slide_c2}
\end{align}
The following technical lemma uses $\Delta_E$ to
bound the volume  covered by the translates of a convex set along a smooth curve. The proof
can be found in Section \ref{sec_proof_of_lemma_slide_b}.
\begin{lemma}
\label{lemma_slide_b}
Let $E \subseteq \Rdst$ be a compact and convex set and
let $\alpha:[0,L] \to \Rdst$
be a rectifiable curve. Let $F \subseteq [0,L]$ be a finite set and consider the set
\begin{align}
E_F := \bigcup_{t \in F} E + \alpha(t).
\end{align}
Then $\abs{E_F} \leq \abs{E} + \ell (\alpha)\Delta_E$.
\end{lemma}

We now prove the main proposition that relates gaps and the path density.

\begin{proposition}
\label{prop:optpdforgap}
Let $E$ be a convex compact set $E \subseteq \Rdst$ with $0 \in
E^\circ$ and let $P  = \{p_i: i \in \bI\}$ be a trajectory set
satisfying condition~\Con{path3}.
If  the translates of $E$ along the trajectories in $P$ cover $\Re^d$, i.e.,
\begin{equation}
\label{eqn:pcovers}
P+E = \bigcup_{\substack{t \in \Re\\  i \in \bI}} E+p_i(t) = \Re^d \, ,
\end{equation}
then
\[
\ell^{-}(P) \geq \frac{1}{\Delta_E} .
\]
\end{proposition}
\begin{proof}
Since $E$ is compact, there exists $R>0$ such that $E \subseteq B_{R}^d$.
Let  $a \geq R$ and $x \in \Rdst$ be arbitrary.
Since $P$ satisfies condition~\Con{path3},
there exists a continuous rectifiable curve $\alpha:[0,1] \to \Rdst$ such that (i)  $\alpha$
contains the entire portion of $P$ inside $B_a^d(x)$ and (ii)
$\length(\alpha)= \clD^P(a,x) + \littleo(a^d)$.

Let us consider the set
\begin{align*}
S = E + \alpha([0,1]).
\end{align*}
We  estimate $|S|$ in two different ways.  Firstly,
if $p_i(t) \notin B_a^d(x)$, then $B_{R}^d(p_i(t))  \cap B_{a-R}^d(x)
= \emptyset$. In view of \eqref{eqn:pcovers}  we have
$$
B_{a-R}(x) \subset \bigcup _{p_i(t) \in B_a^d(x)} E+p_i(t)  \subseteq
S \, .$$
Consequently,
\begin{align}
\label{eq_bound_alpha_a}
\abs{B_{a-R}^d(x)}=  (a-R)^d |B_1^d|  =  a^d |B_1^d|  - O(a^{d-1}) = |B_a^d(x)| - O(a^{d-1}) \leq \abs{S}.
\end{align}
Secondly, since  $S$ is a sum of two compact sets, $S$  is  compact.
Let $\delta \in (0,1)$ and consider the (open) set $(1+\delta)E^\circ$. By Lemma \ref{lemma_conv_1},
$E \subseteq (1+\delta)E^\circ$. Consequently,
\begin{align*}
\sett{(1+\delta)E^\circ + \alpha(t): t \in [0,1]}
\end{align*}
is an open cover of $S$, and there exists
a finite set $F \subseteq [0,1]$ such that
\begin{align*}
S \subseteq \bigcup_{t \in F} (1+\delta)E +\alpha(t).
\end{align*}
Using Lemma \ref{lemma_slide_b} and \eqref{slide_c1},  \eqref{slide_c2} it follows that
\begin{align*}
\abs{S} &\leq \abs{\bigcup_{t \in F} (1+\delta)E+\alpha(t)}
\leq \abs{(1+\delta)E} + \length(\alpha) \Delta_{(1+\delta)E}
\\
&= (1+\delta)^d\abs{E} + \length(\alpha) (1+\delta)^{d-1}\Delta_{E}.
\end{align*}
Combining this estimate with \eqref{eq_bound_alpha_a} we deduce that
\begin{align*}
|B_a^d(x)| \leq O(a^{d-1}) + (1+\delta)^d\abs{E} + \length(\alpha) (1+\delta)^{d-1}\Delta_{E},
\qquad \delta \in (0,1).
\end{align*}
Since this inequality holds for all $\delta >0$ and $o(a^d)$ is
independent of $\delta $ by assumption~\Con{path3}, we obtain
\begin{align*}
|B_a^d(x)| \leq O(a^{d-1}) + \abs{E} + \length(\alpha) \Delta_{E},
\qquad
a \geq R, x \in \Rdst.
\end{align*}
Recalling that $\length(\alpha)= \clD^P(a,x) + \littleo(a^d)$ we obtain
\begin{align*}
|B_a^d(x)| \leq O(a^{d-1}) + \abs{E} + \clD^P(a,x) \Delta_{E}
+ \littleo(a^d)\Delta_{E},
\qquad
a \geq R, x \in \Rdst.
\end{align*}
Therefore,
\begin{align*}
\ell^{-}(P) = \liminf_{a \to \infty} \frac{\inf_{x\in\Re^d} \clD^P(a,x)}{\abs{B_a^d(x)}}
\geq \frac{1}{\Delta_E},
\end{align*}
as claimed.
\end{proof}
\begin{remark}
{\rm In \cite{beu89}, Beurling gave sufficient conditions for a non-uniform collection of points $\Lambda \subseteq
\Rdst$
to form a stable sampling set for the class of bandlimited functions in high dimensions. These are expressed in terms of
a covering condition: $\Lambda + E=\rd $ for a certain convex set $E$ associated with the spectrum support of the
signals. On the other hand, Proposition \ref{prop:optpdforgap} gives a condition on a trajectory that is necessary for
it to contain a sampling set $\Lambda$ satisfying $\Lambda + E=\rd$.}
\end{remark}

We finally prove the main estimate on the density of paths that contain sampling sets with given stability parameters.
\begin{thm}
\label{th:infab}
Let $\Omega \subseteq \Rdst$ be a compact  set with smooth boundary. Then
\begin{equation*}
\inf_{P \in \nomab} \ell^{-}(P) \geq C_d \left( \frac{A |\Omega|}{B
    \sigma (\partial \Omega )} \right)^{d-1},
\end{equation*}
where $C_d$ is a constant that depends only on $d$.

If $\Omega=[-1/2,1/2]^2$, then explicitly
\begin{equation*}
\inf_{P \in \nomab} \ell^{-}(P) \geq \frac{A}{\pi^2\sqrt{2}B}.
\end{equation*}
\end{thm}
\begin{proof}
Since $P$ contains a sampling set with stability parameters $A,B>0$,
Proposition~\ref{prop:gaps} implies that $\Lambda \cap Q_R(x) \neq \emptyset $
for all $x\in \Re ^d$ and with $R = C \frac{B \sigma (\partial\Omega )}{A
  \abs{\Omega}}$. Then $\Lambda + Q_R = \Re ^d$ and thus also $P +
Q_R = \Re ^d$. By Proposition \ref{prop:optpdforgap} we obtain that
$\ell ^-(P) \geq 1/\Delta _{Q_R}$. Since
by \eqref{slide_c2}
\begin{align}
\label{eq_delta_R}
\Delta_{Q_R} = (2R)^{d-1} \Delta_{[-1/2,1/2]^d} = 2^{d-1}
\Delta_{[-1/2,1/2]^d} C^{d-1} \Big(\frac{B \sigma (\partial\Omega)}{A
  \abs{\Omega}}\Big)^{d-1} \, ,
\end{align}
the conclusion follows.
For the case $\Omega=[-1/2,1/2]^2$ we use the exact value $\Delta_{[-1/2,1/2]^2}=\sqrt{2}$
and the explicit estimate for $R$ from Proposition \ref{prop:gaps}:
\begin{align*}
R = \frac{1}{2} + \frac{\pi^2}{4}\frac{B}{A} \leq \frac{\pi^2}{2}\frac{B}{A}.
\end{align*}
\end{proof}

\section{Conclusion}\label{sec:conc}
We have studied the problem of designing trajectories for sampling bandlimited spatial fields using mobile
sensors. We have identified trajectory sets composed of parallel lines that (i) possess minimal path density and
(ii) admit the stable reconstruction of bandlimited fields from measurements taken on these trajectories. We also have
shown that the problem of minimizing the  path density is ill-posed if we allow arbitrary trajectory sets that admit
stable
reconstruction. As a positive result  we have shown  that the problem is well-posed if we restrict  the trajectory sets
to contain a stable sampling set with given stability parameters.

We point out that, for the results presented here, the assumption that the spectrum of the signals is convex is not
essential, but a matter of convenience. Indeed, in most results the convexity of $\Omega$ can be replaced by a suitable
assumption on the regularity of its boundary (eg. Lemma \ref{lemma_sections}). In Theorem \ref{th:parallel} the
convexity of $\Omega$ is used to guarantee that the maximal area of the cross-sections by hyperplanes is attained by a
hyperplane that goes through the origin. For non-convex spectra, a characterization analogous to the one in
Theorem \ref{th:parallel} should consider cross-sections by arbitrary hyperplanes.

This work opens up several possible research directions. One question is whether we can solve the problem
\eqref{eqn:minpdAB} exactly. This would require a tight lower bound on the path density of every trajectory set in
$\nomab$. Another interesting variation concerns trajectory sets
consisting of arbitrary, not necessarily parallel lines and the
necessary path density.

\section*{Acknowledgment}
K.\  Gr\"ochenig was partially supported by National Research Network S106 SISE and by the project P 26273-N25 of the
Austrian Science Fund (FWF). J.\ L.\  Romero gratefully acknowledges support from the project M1586-N25 of the Austrian
Science Fund (FWF) and from an individual Marie Curie fellowship, within the 7th. European Community Framework program,
under grant PIIF-GA-2012-327063. J.\ Unnikrishnan and M.\ Vetterli were  supported by ERC Advanced Investigators Grant:
Sparse Sampling: Theory, Algorithms and Applications – SPARSAM – no. 247006.

\section{Some technical tools and proofs}
\label{sec_tec}
\subsection{Translations and projections of convex sets.}
\begin{lemma}
\label{lemma_slide}
Let $E \subseteq \Rdst$ be a compact convex set and $q\in \Re ^d
\setminus \{0\}$. Then
\begin{align*}
\abs{(E+q) \setminus E} \leq \abs{\clP_{q^\perp } E} \norm{q}_2.
\end{align*}
\end{lemma}
\begin{proof}
  By applying a suitable rotation, we may assume without loss of
  generality  that $q=\alpha e_d=(0,\ldots,0,\alpha)$ for some
  $\alpha > 0$. Then the projection of $E$ onto the hyperplane
  determined by $q$  is simply
$$
P_{q^\perp} E = \{ (x',0) \in \Re ^{d-1} \times \Re : (x', t) \in E\}
\, .
$$
For $x' \in P_{q^\perp} E$ we set
$\tau _- (x') = \min \{ t: (x', t) \in E \}$ and $\tau _+ (x') = \max
\{ t: (x', t) \in E \}$. Since $E$ is compact, the minima and maxima
exist; and since $E$ is convex, the line segments $\{ (x', t): \tau _-
(x') \leq t \leq \tau _+ (x')\}$ are contained in $E$, so that
$$
E = \{ (x',t) \in \Re ^d: x' \in P_{q^\perp} E,  \tau _- (x') \leq t
\leq \tau _+ (x') \} \, .
$$
Consequently
\begin{align*}
&(E+q) \setminus E = (E+\alpha e_d)  \setminus E
\\
& \qquad = \{ (x',t) \in \Re ^d: x' \in
P_{q^\perp} E,  t \in [\tau _- (x') +\alpha ,  \tau _+ (x')+\alpha ]
\setminus [\tau _- (x')  ,  \tau _+ (x')]\} \, ,
\end{align*}
and each fibre over $x'$  has length $\leq \alpha = \|q\|_2$. Now
using Fubini's theorem, we obtain that
\begin{align*}
  |(E+q) \setminus E| &= \int _{\Re ^d } 1_{(E+q) \setminus E} (x', t) \,
  dx' dt \\
&= \int _{P_{q^\perp} E} \int _{\Re} 1_ {[\tau _- (x') +\alpha ,  \tau _+ (x')+\alpha ]
\setminus [\tau _- (x')  ,  \tau _+ (x')]}(t) dt \, dx' \\
&\leq \alpha \int _{P_{q^\perp} E} 1 \, dx' = \alpha |P_{q^\perp} E| =
|P_{q^\perp} E|  \|q\|_2 \, ,
\end{align*}
as claimed.
\end{proof}

\subsection{Spectral gaps for the square. Proof of Proposition \ref{prop:gaps}(b).}
\label{app:gap}
The following proposition - that is part (b) of Proposition \ref{prop:gaps}, restated for convenience -
gives an explicit estimate for the gap of sampling sets for the spectrum $[-1/2,1/2]^d$. Its
proof is inspired by the simple proof of Laudau's necessary conditions for sampling and interpolation given in
\cite{niol12}.

\begin{proposition*}
Let $\Omega:=[-1/2,1/2]^d$ and assume that $\Lambda$ is a sampling set
for $\mathcal{B}_\Omega$ with bounds $A,B$. Then $\Lambda$ intersects
every cube $Q_R(x)= [-R,R]^d+x$,
where
\begin{align}
\label{eq_R}
R = \frac{1}{2}+ \frac{2d}{\pi^2} \frac{B}{A} \Big(\frac{\pi^2}{4}\Big)^d.
\end{align}
\end{proposition*}
\begin{proof}
Since every translation of $\Lambda$ is also a sampling set for $\mathcal{B}_\Omega$ with bounds $A,B$,
it suffices to show that $\Lambda$ intersects $[-R,R]^d$, where $R$ is given by \eqref{eq_R}.
Let $h(x):=\sinc(x)=\Pi_{k=1}^d\frac{\sin(\pi x_k)}{\pi x_k}$, so
$\widehat{h}=1_\Omega$.
We start by noting some facts.

{\bf Claim 1}.
\begin{align*}
A \leq \sum_{\lambda \in \Lambda} \abs{h(\cdot-\lambda)}^2 \leq B.
\end{align*}
\begin{proof}[Proof of the claim]
Note that
\begin{align*}
\sum_{\lambda \in \Lambda} \abs{h(x-\lambda)}^2
=\sum_{\lambda \in \Lambda} \abs{\ip{h(\cdot-\lambda)}{h(\cdot-x)}}^2.
\end{align*}
Since $\sett{h(\cdot-\lambda):\lambda\in\Lambda}$ is a frame with bounds $A,B$
and $\norm{h}_2=1$, the conclusion follows.
\end{proof}

{\bf Claim 2}.
\begin{align*}
\# (\Lambda\cap Q_{1/2}(x)) \leq \Big(\frac{\pi^2}{4}\Big)^d \,  B.
\end{align*}
\begin{proof}[Proof of the claim]
Since $\frac{\sin \pi t}{\pi t} \geq 2/\pi $ for $|t| \leq 1/2$, we
have $h(x) \geq (2/\pi )^d$ for $x\in [-1/2,1/2]^d =
Q_{1/2}(0)$. Therefore we obtain
\begin{align*}
\Big(\frac{4}{\pi^2}\Big)^d \# (\Lambda\cap Q_{1/2}(x))
\leq \sum_{\lambda \in \Lambda} \abs{h(\lambda-x)}^2 \leq B \norm{h(\cdot-x)}_2^2 = B.
\end{align*}
\end{proof}

{\bf Claim 3}.
\begin{align*}
\int_{\Rdst \setminus Q_{r}(0)} \abs{h(x)}^2 dx \leq \frac{2d}{\pi^2 r},
 \qquad \forall r>0 \, .
\end{align*}
\begin{proof}[Proof of the claim]
Since $\sinc(x)=\sinc(x_1)\dots\sinc(x_d)$ and each one-dimensional sinc
is normalized in $L^2$,  we estimate
\begin{align*}
\int_{x \in \Rdst, \abs{x}_\infty > r} \abs{\sinc(x)}^2 dx
&\leq \sum_{k=1}^d
\int_{x \in \Rdst, \abs{x_k}> r} \abs{\sinc(x)}^2 dx
\\
&= \sum_{k=1}^d \int_{t \in \Rst:\abs{t} > r} \abs{\sinc(t)}^2 dt
\\
&\leq 2d \int_r^\infty \frac{1}{(\pi t)^2} dt = \frac{2d}{\pi^2r}.
\end{align*}
\end{proof}
Combining the claims we get
\begin{align*}
A = A \abs{Q_{1/2}(0)} &\leq \int_{Q_{1/2}(0)} \sum_{\lambda \in \Lambda} \abs{h(x-\lambda)}^2 dx
= \int_{\Rdst} \abs{h(x)}^2 \sum_{\lambda \in \Lambda}
1_{Q_{1/2}(\lambda)}(x) dx 
\\
&= \int_{\Rdst} \abs{h(x)}^2 \# (\Lambda\cap Q_{1/2}(-x)) dx 
\\
&\leq \Big(\frac{\pi^2}{4}\Big)^d  B \int_{\bigcup_{\lambda \in \Lambda} Q_{1/2}(\lambda)} \abs{h(x)}^2 dx.
\end{align*}
Now assume that $\Lambda \cap Q_R(0) = \emptyset$. Then, for every $\lambda \in \Lambda$,
$Q_{1/2}(\lambda) \cap Q_{R-1/2}(0) = \emptyset$. Therefore,
\begin{align*}
A &\leq \Big(\frac{\pi^2}{4}\Big)^d
B \int_{\Rdst \setminus Q_{R-1/2}(0)} \abs{h(x)}^2 dx
\leq \frac{2d}{\pi^2} B \Big(\frac{\pi^2}{4}\Big)^d (R-1/2)^{-1}.
\end{align*}
Hence, $R\leq \bound$.

This means that $\Lambda$ must intersect $[-R,R]^d$ if $R>\bound$, as desired.
(Since $\Lambda$ is closed, it also follows that $\Lambda$ intersects $[-R,R]^d$
for $R=\bound$).
\end{proof}

\subsection{Continuity of sections of convex sets. Proof of Lemma \ref{lemma_sections}.}
\label{sec_proof_lemma_sections}
\begin{proof}[Proof of Lemma \ref{lemma_sections}]
Without loss of generality let us assume that $\Omega_t \not= \emptyset$.
Suppose that the conclusion does not hold. Then there exists a sequence of real numbers
$\{t_n: n \geq 1\}$ such that $t_n \longrightarrow t$ and
\begin{align*}
\Omega_{t_n} \not\subseteq \Omega_t + B_\varepsilon.
\end{align*}
Hence there exist points $x_n \in \Omega_{t_n}$ such that
\begin{align}
\label{eq_dist}
r_n:=d(x_n, \Omega_t)=\inf\{\abs{x_n-y}: y \in \Omega_t\} \geq \varepsilon.
\end{align}
Since $\Omega_t$ is closed, there exists
$y_n \in \Omega_t$ such that $\abs{x_n-y_n} = r_n$.

Consider the sequences $\{(x_n,t_n): n \geq 1\}, \{(y_n,t): n \geq 1\} \subseteq \Omega$.
By passing to subsequences we may assume that both of them are convergent:
\begin{align*}
&(x_n,t_n) \longrightarrow (x,t),
\\
&(y_n,t) \longrightarrow (y,t).
\end{align*}
Hence, $x,y \in \Omega_t$. In addition,
by \eqref{eq_dist}, $r:=\abs{x-y}=\lim_n \abs{x_n-y_n}=\lim_n r_n \geq \varepsilon>0$.

Since $\Omega$ is convex,  so is $\Omega_t$. Consequently, $z=(x+y)/2 \in \Omega_t$.
Let us estimate
\begin{align*}
\abs{x_n-z}\longrightarrow \abs{x-z}=r/2 < r=\lim_n r_n.
\end{align*}
Therefore, there exist $n \in \Nst$ such that $\abs{x_n-z} < r_n$. Since $z \in \Omega_t$,
this contradicts the fact that $r_n=d(x_n, \Omega_t)$.
\end{proof}

\subsection{Sliding convex sets. Proof of Lemma \ref{lemma_slide_b}.}
\label{sec_proof_of_lemma_slide_b}
\begin{proof}[Proof of Lemma \ref{lemma_slide_b}]
Let us enumerate the points of $F$ as
$0 \leq t_0 < \ldots < t_N \leq L$. Without loss of generality we further assume that
$\alpha(t_k) \not= \alpha(t_j)$, for $k\not=j$ (Indeed, if $\alpha(t_k) = \alpha(t_j)$, for some $k \not= j$,
then we may remove $t_j$ from the set $F$ without altering the set $E_F$.). Let us consider the sets
$E_k:=E+\alpha(t_k)$.

For $1 \leq k \leq N$, let $q_k := \alpha(t_{k-1})-\alpha(t_k)$. By Lemma \ref{lemma_slide}, it follows that
\begin{align*}
\abs{E_k \setminus E_{k-1}} &\leq \abs{\clP_{q_k^\perp } (E)} \| \alpha(t_k)-\alpha(t_{k-1})\|_2.
\end{align*}
Since $\alpha(t_{k-1}) \not= \alpha(t_k)$, $q_k \not=0$, for all $k$. Hence, considering
the vectors $q'_k:=\norm{q_k}^{-1} q_k$ we see that
\begin{align*}
\abs{\clP_{q_k^\perp } (E)}
= \abs{\clP_{(q'_k)^\perp } (E)} \leq
\Delta_{E}.
\end{align*}
Therefore,
\begin{align*}
\abs{E_k \setminus E_{k-1}} \leq \Delta_E \| \alpha(t_k)-\alpha(t_{k-1})\|_2,
\qquad 1 \leq k \leq N.
\end{align*}
Let us decompose $E_F$ as
\begin{align*}
E_F := E_0 \cup \bigcup_{k=1}^{N} (E_k \setminus E_{k-1}).
\end{align*}
Hence,
\begin{align*}
\abs{E_F} &\leq \abs{E_0} + \sum_{k=1}^{N} \abs{E_k \setminus E_{k-1}}
\\
&\leq \abs{E} + \sum_{k=1}^{N} \Delta_E \|\alpha(t_k)-\alpha(t_{k-1})\|_2
\leq \abs{E} + \Delta_E \length(\alpha).
\end{align*}
as claimed.
\end{proof}
\bibliographystyle{abbrv}

\end{document}